\documentclass{amsart}

	\usepackage{amssymb, amsmath, amsthm}
\numberwithin{equation}{section}
\usepackage{cite}                   
\usepackage{paralist}               
\usepackage[margin=1in]{geometry}
\usepackage{appendix}    
\usepackage[final]{graphicx}
\usepackage{enumitem}
\usepackage{subcaption} 
\usepackage{float}
\usepackage{mathrsfs}
\usepackage{array}\usepackage[bookmarks,bookmarksnumbered,bookmarksopen,breaklinks,linktocpage]{hyperref}
\usepackage{nomencl}

\newtheorem{Proposition}{Proposition}[section]
\newtheorem{Lemma}{Lemma}[section]

\theoremstyle{definition}           

\theoremstyle{remark}

\begin{document}

\title[Selection of Saffman-Taylor fingers in a tapered Hele-Shaw cel]{On the selection of Saffman-Taylor fingers in a tapered Hele-Shaw cell}

\author{Dipa Ghosh}
\address{Department of Mathematics, Indian Institute of Technology Guwahati, Guwahati - 781039, Assam, India}
\email{dipamath@iitg.ac.in}

\author{Satyajit Pramanik}
\address{Department of Mathematics, Indian Institute of Technology Guwahati, Guwahati - 781039, Assam, India}
\email{satyajitp@iitg.ac.in (Corresponding author)}





\begin{abstract}
We present an analytical study for predicting the finger width of the Saffman-Taylor finger in a tapered Hele-Shaw cell. We consider a rectilinear geometry with a constant depth gradient and apply analytical techniques of singular perturbation analysis and WKB approximation to derive an expression for the finger selection mechanism for such tapered Hele-Shaw cells with small depth gradients. We establish 
\[ \Lambda - \frac{1}{2} \sim f(\alpha) Ca_m^{2/3} \quad \mbox{as} \quad Ca_m \rightarrow 0, \;\;\; \mbox{and} \;\;\;  \lvert \alpha \rvert \ll 1.\]
Here, $\Lambda$ is the dimensionless finger width, $Ca_m$ denotes the modified Capillary parameter, and $f(\alpha)$ is a linear function of the gap gradient $\alpha$, such that $f(\alpha = 0) = 1$ recovering the results of parallel Hele-Shaw cell (Hong and Langer \cite{hong1986analytic}, Combescot \emph{et al.} \cite{Combescot1986}, Shraiman \cite{shraiman1986velocity}). Our findings indicate that the Hele-Shaw cell gap gradient plays a crucial role in determining $\Lambda$, allowing for control over fingering instabilities such that the single-finger steady state can be stabilised or destabilised depending on the sign of the gradient, compared to the standard Hele-Shaw cell. The theoretical estimates reveal excellent agreement with experimental finger-width data and predictions from linear stability analyses. 
\end{abstract}

\maketitle

\section{Introduction}


Viscous fingering (VF) instability is a fundamental interfacial instability in hydrodynamic systems and represents a canonical mechanism for spontaneous pattern formation in viscous flows \cite{hill1952channeling,saffman1958penetration,chuoke1959instability}. It has been extensively investigated within the framework of interfacial hydrodynamics \cite{hill1952channeling, homsy1987viscous}, nonlinear stability theory \cite{alvarez2004nonlinear}, and pattern selection \cite{mclean1981effect, hong1986analytic, Combescot1986, hong1987pattern} across multiple disciplines. The instability was first demonstrated experimentally in 1958 by Saffman and Taylor \cite{saffman1958penetration} in a rectilinear Hele–Shaw cell, composed of two parallel glass plates, separated by a narrow uniform gap. Subsequently, the problem has been extensively studied in radial Hele–Shaw geometries \cite{thome1989saffman}. In the radial configuration, Hele-Shaw flows take place when the less viscous fluid is injected at a constant injection rate at the centre of the cell, and drives the more viscous fluid radially outwards. The gap-averaged velocity in a Hele–Shaw cell \cite{saffman1986viscous} is mathematically equivalent to a two-dimensional viscous flow through a porous medium and obeys Darcy’s law; hence, it serves as a convenient laboratory analogue for porous media displacements \cite{hill1952channeling,saffman1958penetration,chuoke1959instability,mclean1981effect, homsy1987viscous}. Since then, this phenomenon—commonly referred to as the Saffman–Taylor instability or pore-scale viscous fingering—has formed the basis of a vast body of theoretical, numerical, and experimental research on interfacial instabilities. 

In the classical Saffman–Taylor experiment \cite{saffman1958penetration}, a low-viscosity Newtonian fluid (e.g., water) is injected into a cell initially filled with a more viscous Newtonian fluid (e.g., oil). Under these conditions, the initially planar interface becomes hydrodynamically unstable and develops finger-like protrusions of the invading phase into the displaced one, leading to the breakdown of the initially straight interface. At sufficiently high velocities, a single finger emerges and is stabilized by surface tension, occupying nearly half of the channel width ($W$)  \cite{saffman1958penetration}, a result later supported by the theoretical analysis of Pitts \cite{pitts1980penetration} and many others \cite{tabeling1986film,tabeling1987experimental}. However, their analytic solution without surface tension does not restrict the value of $\Lambda$ \cite{mclean1981effect}. The significance of the surface tension to the shape selection procedure was analytically highlighted much later by Hong and Langer in 1986 \cite{hong1986analytic} and similar parallel works by others \cite{Combescot1986,shraiman1986velocity,tanveer1986effect,bensimon1986viscous}, who showed that the surface tension appeared in the problem as a singular perturbative term leading to a solvability condition at the fingertip, thereby isolating a particular value from the continuum of solutions. However, the analytical solution was derived for Newtonian-Newtonian fluid displacements. In this regard, Bansal et al. \cite{bansal2023selection} recently derived an analytical expression for the non-Newtonian finger profile in a parallel rectangular Hele–Shaw cell, capturing both finger-thinning and finger-widening regimes. 
\subsection{Applications and importance}\label{subsec:AI}
Although viscous fingering is a paradigmatic instability in pattern formation, it arises widely in industrial and geological processes, such as enhanced oil recovery (EOR) \cite{gorell1983theory ,lake1988enhanced, ali1996promise, pinilla2021experimental, vishnudas2017comprehensive} and CO$_2$ geological sequestration \cite{huppert2014fluid}, where it is often undesirable. In EOR, water is injected into porous oil reservoirs to displace oil; however, due to the Saffman–Taylor instability, fast-growing water fingers can reach the production well prematurely, resulting in the extraction of mostly water rather than oil and a severe reduction in recovery efficiency. This makes viscous fingering a serious morphogenetic nuisance and a major obstacle in petroleum engineering. Given the growing global demand for energy, therefore, improving oil recovery from porous reservoirs has become increasingly important \cite{al2018numerical, ali1996promise, pinilla2021experimental, talebian2014foam, daripa2004optimal, daripa2012numerical}. Because of situations like this, much attention has been devoted to devising strategies for controlling, and eventually stabilizing, the growth of such patterns in immiscible \cite{gorell1983theory ,homsy1987viscous,pasa1996existence} and miscible displacements \cite{tan1988simulation,fernandez2002density,pramanik2015effect} in porous media.
\subsection{Control of Viscous Fingering Instabilities}\label{subsec:CII}
Many studies have shown that uncontrolled viscous fingering can severely compromise the efficiency of oil recovery operations, thereby motivating the development of strategies to mitigate such instabilities. One such approach is stabilization through the introduction of a depth gradient in confined geometries \cite{maxworthy2002effect}. Early investigations in the context of thin films between rollers and spreaders examined the influence of depth gradients on the Saffman–Taylor instability \cite{pearson1960instability, pitts1961flow}. More recently, this problem has attracted renewed attention owing to the close analogy between flow in a Hele–Shaw cell and flow through porous media \cite{homsy1987viscous, zimmerman1996hydraulic, bear2013dynamics}. In this context, Al-Housseiny et al. \cite{al2012control, al2013controlling} reinterpreted the linear stability analysis in the presence of a small gap gradient—i.e., in an angled Hele–Shaw cell—as a viable stabilization and control mechanism. Earlier experimental work by Zhao et al. \cite{zhao1992perturbing} demonstrated that Saffman–Taylor fingers can be either stabilized or destabilized depending on the sign of the gap gradient, while Dias and Miranda \cite{dias2010finger} extended linear stability theory to predict finger tip-splitting in the presence of variations in the flow passage.

Subsequently, a large body of work has investigated interfacial instabilities in a variety of non-standard Hele–Shaw configurations and flow geometries. In tapered and angled cells, where the plate spacing varies along the flow direction, depth gradients modify the base flow and the linear growth rates of perturbations, providing a powerful mechanism to stabilize or destabilize viscous fingering \cite{dias2010finger,al2012control, al2013controlling, dias2013taper, lu2020computational}. In angled Hele-Shaw cells, theoretical and numerical study on the interfacial instability has been done across a range of capillary numbers ($Ca$) (\cite{lu2020computational} and reference therein). In elastic or compliant Hele–Shaw cells, deformation of one or both plates leads to a dynamically evolving gap, introducing fluid–structure interaction that strongly influences finger morphology and stability \cite{al2013two, pihler2012suppression, pihler2013modelling}. Numerous studies have used exponential asymptotic to derive selection mechanism \cite{chapman1999role, lustri2020selection, chapman2023role} and also for the case of divergent flow within a wedge shaped Hele-Shaw cell \cite{andersen2024selection, andersen2025computations}.  Wedge-shaped and divergent geometries, in which the flow cross-section increases downstream, generate spatially varying velocities and pressure gradients that affect finger selection and tip-splitting dynamics \cite{amar1991viscous, andersen2024selection, andersen2025computations}. Heterogeneous or structured Hele–Shaw cells, designed to mimic porous-media disorder, introduce spatial variations in permeability or gap thickness, leading to complex invasion pathways and non-classical fingering patterns \cite{hu2016numerical, jackson2017stability, bongrand2018manipulation}. In a lifting Hele-Shaw flow for Newtonian (due to inertia) \cite{anjos2017inertia} as well as non-Newtonian fluids \cite{kabiraj2003finger} in a time-dependent manner is also analysed. Finally, time-dependent and dynamically controlled geometries enable active modulation of the gap or forcing conditions, offering additional routes for controlling the growth and evolution of interfacial instabilities \cite{morrow2019numerical, morrow2021review}. 
Unlike elastic boundaries that dynamically reshape the flow passage and affect stability, external electric fields provide an additional body force that alters interfacial stresses and electrokinetic transport, enabling active control of finger evolution \cite{mirzadeh2017electrokinetic}. Time-dependent injection rates have been shown to regulate the amplification of unstable modes by modulating the transient base state \cite{dias2012minimization, gin2015stability, gin2021time}, and multilayer flow configurations can suppress fingering by distributing viscosity contrasts across multiple interfaces \cite{daripa2012universal}. Material-based strategies include tailoring permeability through porous-medium structure \cite{tan1992viscous}, adjusting the viscosity ratio between fluids \cite{anjos2017radial}, and introducing suspended particles to modify effective rheology and interfacial mobility \cite{xu2016particle, hooshanginejad2019stability}. Additional control can be achieved by applying external forces via rotation or magnetic fields, which modify the pressure and force balance at the interface \cite{jackson2003controlling, alvarez2003systematic,miranda2005viscosity, anjos2018finger}, as well as by targeted control of miscible fingering in radial geometries \cite{chen2010controlling, sharma2020control}. Collectively, these approaches demonstrate that geometric, dynamical, and material modifications can profoundly alter the underlying spectral problem governing the stability of the fluid–fluid interface and the growth of interfacial disturbances. 

\subsection{Goals and outline of this work}
In the present work, we connect, for the first time, the classical solvability theory of viscous fingering with the geometric control paradigm of tapered Hele–Shaw cells. While solvability and WKB analyses have successfully explained finger selection in parallel cells \cite{hong1986analytic, corvera1995anisotropic, bansal2023selection}, adding depth gradients have been shown to influence linear stability and finger dynamics strongly \cite{zhao1992perturbing, dias2010finger,al2012control, al2013controlling}. These two frameworks have remained theoretically disconnected. Motivated by these observations, we seek to establish a unified theoretical description that incorporates geometric tapering into the solvability framework. Here, we develop a systematic solvability formulation for Saffman–Taylor fingering in a tapered geometry, demonstrating how a depth gradient enters the selection problem as a new control parameter. We show that even weak tapering breaks the degeneracy of the classical solutions and shifts the selected finger width through a modified capillary–geometry balance, thereby enabling geometric control of nonlinear pattern selection.

This study focuses on two main objectives. First, it provides a clear and systematic derivation of the solvability condition governing pattern selection in the Saffman–Taylor problem within a tapered Hele–Shaw cell, following the framework developed in earlier works \cite{corvera1995anisotropic, bansal2023selection}, \cite{zhao1992perturbing,al2012control,al2013controlling}. Second, it demonstrates that the selected finger width $\Lambda$, which depends on the nondimensional depth gradient and the modified capillary parameter $Ca_m$, is sensitive to the imposed depth variation. A slight departure from parallel geometry results in a noticeable shift in finger selection, consistent with experimental findings and linear stability analysis \cite{al2012control, al2013controlling, zhao1992perturbing}. 

The paper is organised as follows. In Section \ref{sec:PF}, we derive the viscous fingering equations in the physical plane and transform them into the potential and conformal planes, leading to a linear inhomogeneous integro-differential equation. This equation is then analysed in Section \ref{sec:Solvability} using the solvability theory and the WKB approximation to obtain the main analytical results. Section \ref{sec:discussion} presents and validates these results through comparisons with in silico simulations and in vitro experiments reported in the literature, followed by the concluding remarks. 

\section{\label{sec:PF}Problem formulation}

A conventional Hele–Shaw cell consists of two parallel plates separated by a small gap $h_0$ and filled with a viscous fluid $F_1$ with viscosity, $\mu$. In a rectilinear configuration, this immiscible fluid $F_1$ is displaced (at constant velocity) by an invading, immiscible fluid $F_2$, with surface tension $\nu$ acting at the interface.

\begin{figure}[!htbp]
    \centering
    \includegraphics[width=\textwidth]{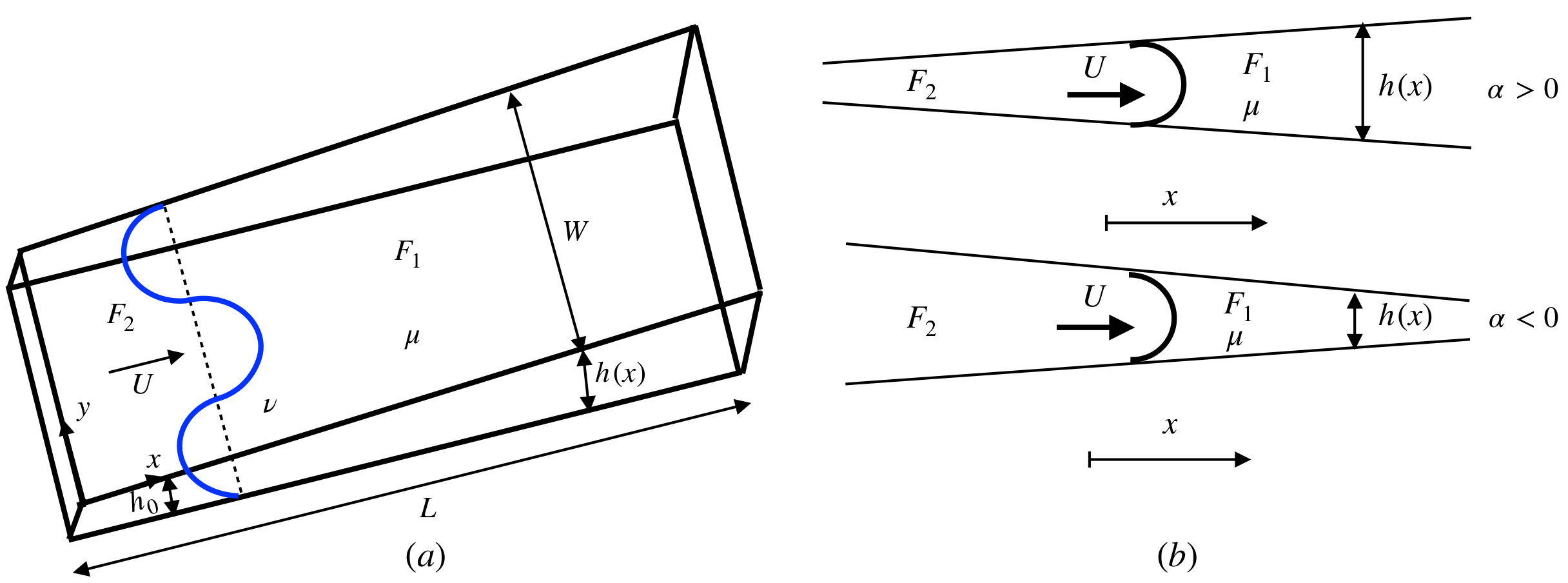}
    \caption{(a) A schematic illustrating the top view of the fluid displacement in a rectilinear tapered Hele-Shaw cell. (b) Schematics of the side view of the cell illustrating a positive and a negative depth gradient along the flow direction.}
    \label{fig:1}
\end{figure}
\vspace{-0.5cm}

In the present study, we examine the displacement of one fluid by another in a Hele–Shaw cell where the plates are not parallel; instead, they are separated by a slight depth gradient $\alpha$, as illustrated in Fig. \ref{fig:1}$(a)$. Because the plates are nearly horizontal $(\lvert \alpha \rvert \ll 1)$, gravitational effects can be ignored relative to viscous forces. The cell has a depth profile  $h(x) = h_0 + \alpha x$, a width $W$, and a length $L$, such that $L \gg W  \gg  h(x)$, representing either a slowly diverging $\alpha > 0$ or converging $(\alpha < 0)$ geometry along the flow direction (see Fig. \ref{fig:1}$(b)$). In this study, we seek to identify the effect of a small depth gradient $\lvert \alpha \rvert \ll 1$ on the stability of a uniform interface, which is located at $x=0$ and is advancing in the Hele-Shaw cell at a velocity $U$.
\subsection{Derivation of the viscous fingering equations in the physical plane} \label{subsec:VFE}
\begin{figure}[h]
    \centering
    \includegraphics[width=0.7\textwidth]{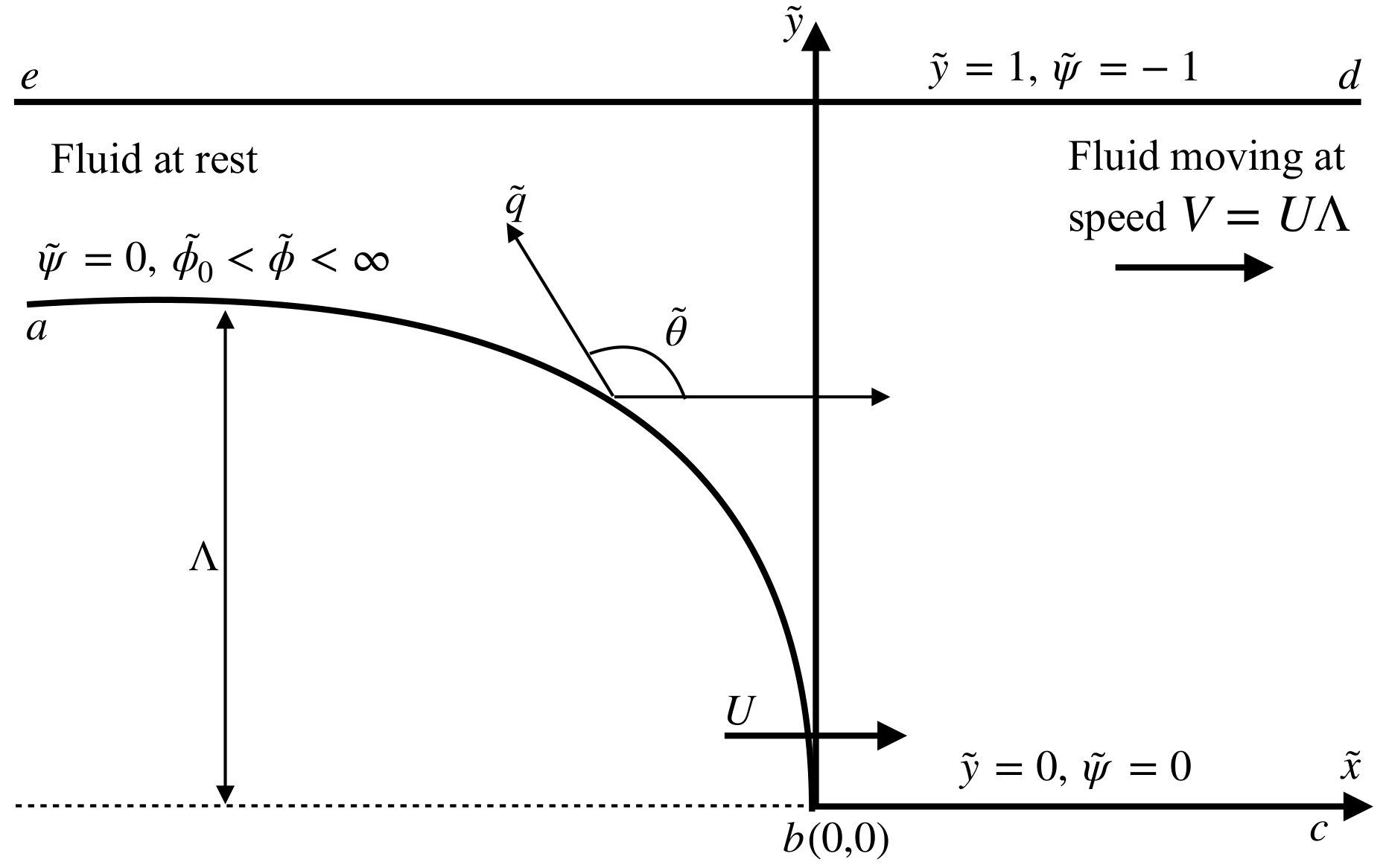}
    \caption{A schematic diagram of a Saffman-Taylor finger in a laboratory frame, assumed symmetric about the $\tilde{x}-$axis (Top view)}
    \label{fig:2}
\end{figure}
\vspace{0.5cm}
Since $L \gg W  \gg  h(x)$, the slow flow of an incompressible fluid in a homogeneous porous medium is governed by two-dimensional Darcy's law and the continuity equation modified to account for the depth gradient
\begin{subequations}
  \begin{align}
     \Vec{u} = -\frac{h^2(x)}{12\mu}\Vec{{\nabla}} p ,
      \label{eqn:1a} \\
    \Vec{\nabla} \cdot (h(x)\Vec{u})=0 ,
     \label{eqn:1b} 
  \end{align}
  \label{eqn:1} 
\end{subequations}
where $p$ and $\Vec{u}=(u,v)$ are, respectively, the depth-averaged pressure and velocity vector. Substituting Eq. \eqref{eqn:1a} into \eqref{eqn:1b}, the flows in the two fluids are described by 
\begin{align}
    \Vec{\nabla}^2 p +\frac{3\alpha}{h(x)}\frac{\partial p}{\partial x} = 0.
    \label{eqn:2}
\end{align}
In the limit where $\lvert \alpha \rvert \ll 1$, the terms of $\mathbf{O}(\alpha^2)$ can be neglected and thus Eq. \eqref{eqn:2} becomes a PDE with constant coefficients,
\begin{align}
    \frac{\partial^2 p}{\partial x^2} + \frac{\partial^2 p}{\partial y^2} + \frac{3\alpha}{h_0}\frac{\partial p}{\partial x} = 0.
    \label{eqn:3}
\end{align}
The interface conditions are described in terms of the continuity of the velocity and the pressure jump across the interface. The former requires that the normal velocity of the two fluids is equal at the interface, 
\begin{align}
    \Vec{u}_{F_1} \cdot \boldsymbol{n}_1  = \Vec{u}_{F_2} \cdot \boldsymbol{n}_2,
    \label{eqn:4}
\end{align}
where $\boldsymbol{n}_1$ and $\boldsymbol{n}_2$ denote the outward unit normal at the interface in the $F_1$ and $F_2$, respectively \cite{corvera1995steady, corvera1995saffman}. In our problem, we introduce surface tension, as it provides a natural length scale and bounds the wavelength of perturbations that could make the interface unstable. Therefore, the second condition is the Gibbs–Thomson condition, which states that the pressure at the interface decreases in proportion to the local curvature, i.e., 
\begin{align}
    -p = \nu \kappa + \frac{2\nu}{h},
    \label{eqn:5}
\end{align}
where $\kappa = 1/R$ is the lateral curvature, where $R$ is the radius of curvature. The second part of the dynamic boundary condition comes from the curvature ($1/h$) due to the depth of the HS cell. It should be noted that in our case, the fluid $F_1$ is perfectly wetting. Assuming $x-$direction as the direction of the flow, we have $\displaystyle \frac{\partial p}{\partial y} = 0$ at the side walls of the channel $y = \pm w$, where $w$ is half of the channel width $W$ ($w = W/2$). This is because the viscous fluid cannot penetrate the rigid walls forming the channel. Far away from the finger, the flow becomes uniform and unidirectional along the length of the cell, i.e., $\Vec{u} = (U, 0)$, and $\displaystyle \left. \frac{\partial p}{\partial x}\right\rvert_{x \rightarrow \infty} = -\frac{12\mu}{h^2} U$.
Eqs. \eqref{eqn:1a} -- \eqref{eqn:3} together with the boundary conditions \eqref{eqn:4} and \eqref{eqn:5} constitute the equations of motion of the problem. 

\subsection{Transformation of the tapered cell into an infinite strip} \label{subsec:conformal}

\begin{figure}[h]
    \centering
    \includegraphics[width=\textwidth]{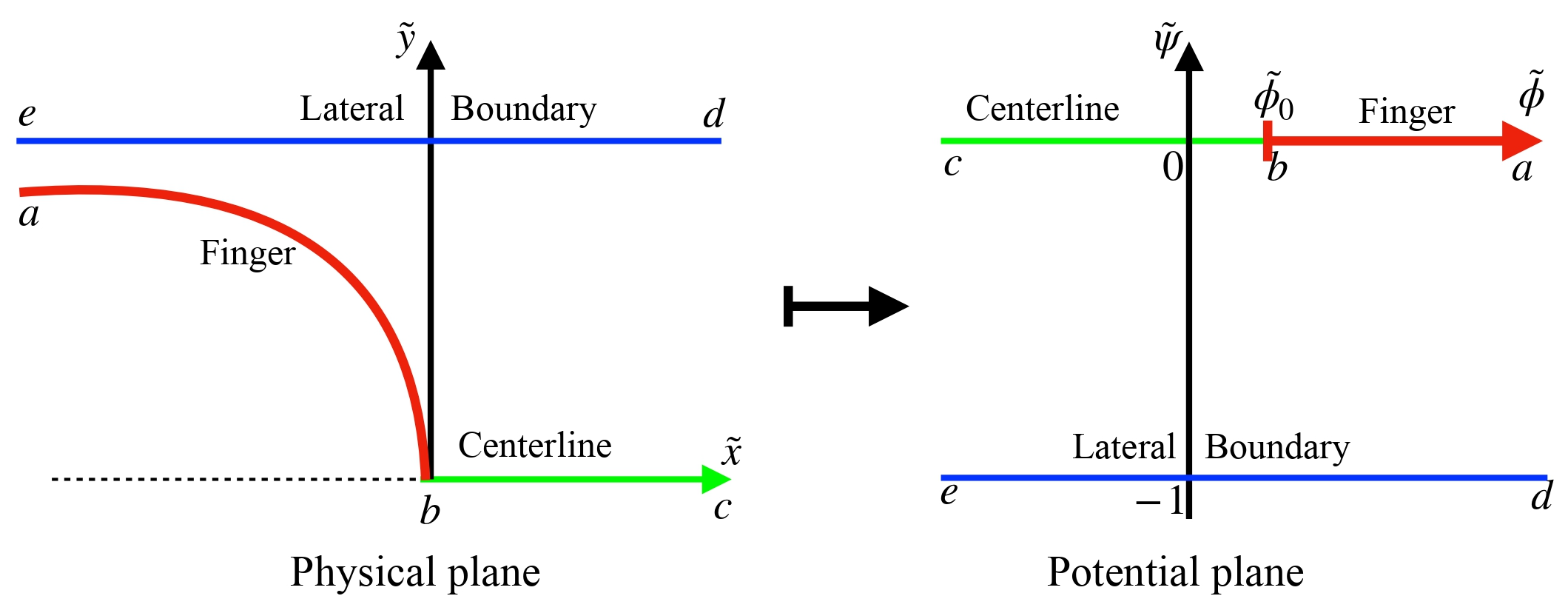}
    \caption{A schematic diagram of a Saffman-Taylor finger from the physical plane to the potential plane}
    \label{fig:3}
\end{figure}
We define $\displaystyle \phi =- \frac{h^3}{12\mu}p$, where $p = p_{F_1} - p_{F_2}$ such that 
\begin{align}
    \Vec{u} =\frac{\Vec{\nabla} \phi}{h}.
    \label{eqn:6}
\end{align}
The scalar function $\phi$ can be interpreted as a velocity potential that satisfies Darcy's law \eqref{eqn:1a} in the leading order. The deviation from Darcy's law is $\mathbf{O}(\alpha)$. Here, we discuss finger selection in a tapered HS cell due to this geometric perturbation $\alpha$. 
Since the flow is two-dimensional and obeys Laplace's equation
\begin{align}
    \Vec{\nabla}^2 \phi =0,
    \label{eqn:8}
\end{align}
the complex potential $\mathcal{G} = \phi + i\psi$ is an analytic function of $z=x+iy$. In these variables, the thermodynamic boundary condition is given by
\begin{align}
    \phi = \frac{h^3}{12\mu}\left(\frac{\nu}{R}+\frac{2\nu}{h}\right).
    \label{eqn:9}
\end{align}
From Eq. \eqref{eqn:4}, the kinematic boundary condition is given by
\begin{align}
    \frac{1}{h}\frac{\partial \phi}{\partial n} = U \sin{\tilde{\theta}},
    \label{eqn:10}
\end{align}
where $\tilde{\theta}$ is the angle between the $x-$axis and $\boldsymbol{n}$ denotes the outward normal from the fluid $F_2$ to fluid $F_1$. The interface conditions given by Eqs. \eqref{eqn:9} and \eqref{eqn:10} are on the advancing finger. On the walls, the boundary conditions are given by
\begin{align}
\left.\begin{aligned}
    \phi \to U\Lambda x w ~\quad \text{as} \quad x \to \infty, -w<y<w,\\
    \phi \to 0 \quad \text{as} \quad x \to -\infty, w\Lambda<|y|<w.
    \end{aligned}
    \right\}
    \label{eqn:11}
\end{align}
Here, we use the fact that $h(x) \ll w$. The system is governed by Laplace’s equation \eqref{eqn:8} together with the boundary conditions \eqref{eqn:9}- \eqref{eqn:11}.

Next, we choose to work in a frame of reference moving with the interface (say, located at $x_0 = Ut$) such that the origin is located at the fingertip and introduce the dimensionless parameters, using $w$ and $(1-\Lambda)Uw$ as length and velocity scales, respectively, as 
\begin{align}
\left.\begin{aligned}
    \Tilde{x}=\frac{x-x_0}{w},~~~ \Tilde{y}=\frac{y}{w}, ~~~\frac{1}{\Tilde{R}}=\frac{w}{R}, ~~~\Tilde{h}=\frac{h}{h_0}, \\~~~\Tilde{\phi}=\frac{\phi-Uxw}{(1-\Lambda)Uw^2}, ~~~\Tilde{\psi} = \frac{\psi - Uyw}{(1-\Lambda)Uw^2}. ~~~~~~~~
    \end{aligned}
    \right\}
    \label{eqn:12}
\end{align}
The free-boundary potential problem is illustrated in Fig. \ref{fig:3}. We assume a symmetric finger ($ab$) about the centerline of the channel --- the line occupying the $x-$axis in front of the finger is denoted by $bc$, while the cell boundary is denoted by $de$. In the complex plane $\tilde{\mathcal{G}} = \tilde{\phi}+i\tilde{\psi}$, the lateral cell boundary corresponds to the streamline $\tilde{\psi} = -1$, the central line occupies the section of the streamline $\tilde{\psi} = 0$ ($ -\infty < \tilde{\phi} < \tilde{\phi_0}$), the finger is the section of the streamline $\tilde{\psi} = 0$ ($\tilde{\phi_0} < \tilde{\phi} < \infty$) and 
\begin{align}
     \Tilde{\phi}= \frac{\nu h_0^3}{12\mu (1-\Lambda)Uw^2}{\tilde{h}}^3 \left(\frac{1}{w\tilde{R}}+\frac{2}{h_0\tilde{h}}\right) - \frac{\tilde{x} + \frac{x_0}{w}}{1-\Lambda} . 
    \label{eqn:13}
\end{align}
\begin{figure}[h]
    \centering
    \includegraphics[width=\textwidth]{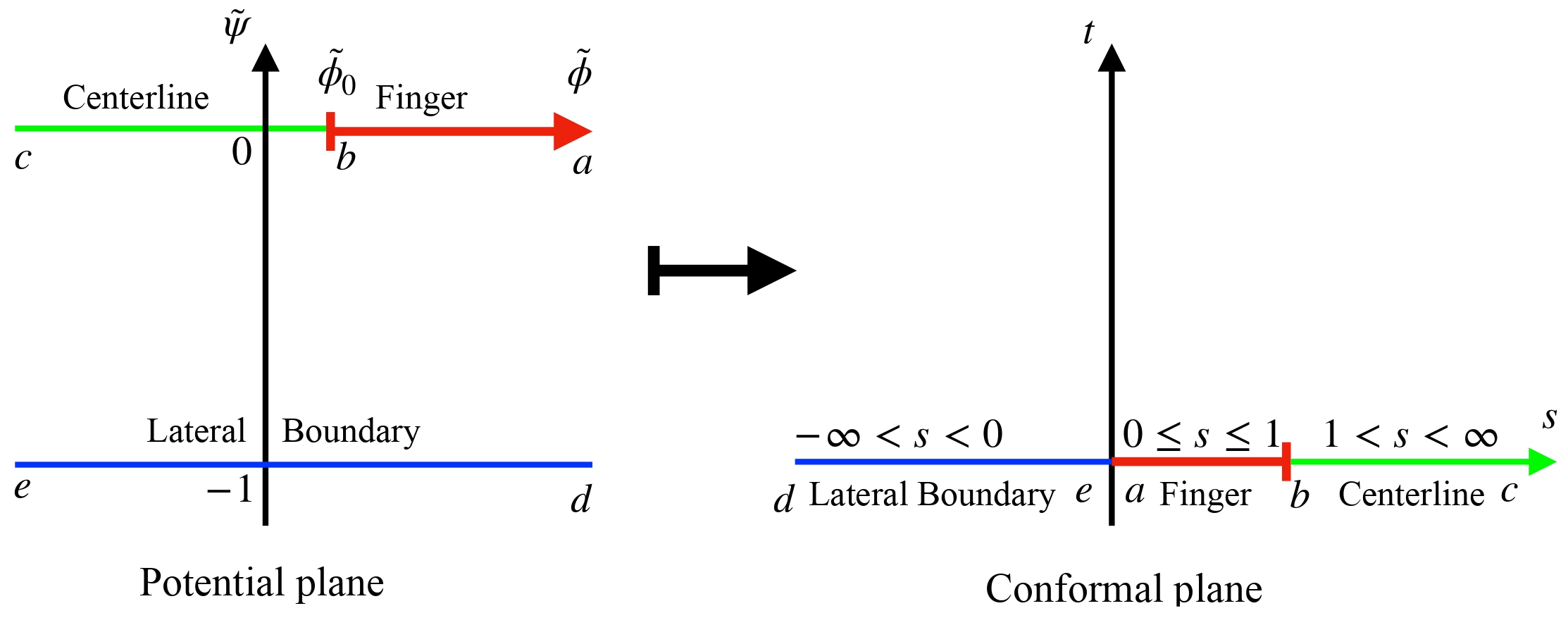}
    \caption{A schematic diagram of a Saffman-Taylor finger from potential plane to conformal plane}
    \label{fig:4}
\end{figure}

A natural (dimensionless) complex velocity of magnitude $\tilde{q}$ is given by $\tilde{u}-i\tilde{v}$. We write this complex velocity $\tilde{u}-i\tilde{v}$ in polar coordinates as $\tilde{q}e^{-i\tilde{\theta}}$ and study the logarithm of this function, i.e., $-\ln \tilde{q}+i\tilde{\theta}$ in the complex plane (using the Kirchhoff-Helmholtz free streamline theory), so that it is given by the conformal map of the potential plane,
\begin{align}
    \sigma = s +it = e^{-(\tilde{\mathcal{G}} - \tilde{\phi_0})\pi}.
    \label{eqn:14}
\end{align}
This conformal transformation maps the potential plane into the upper half $\sigma$-plane (Fig. \ref{fig:4}) with the cell boundary $(de)$, the interface $(ab)$ and the centerline $(bc)$ mapped into the real $s$-axis in the following way: the interface occupies the section $0<s<1$, the centerline the section $1<s<\infty$ and the cell boundary the negative real axis $-\infty < s < 0$. On $bc$ and $de$, $\tilde{\theta} = \pi$. 

Applying Cauchy principal-value integral theorem to the logarithm of the complex velocity, we obtain \cite{mclean1981effect}
\begin{align}
    \ln (\tilde{q}(s)) = -\frac{1}{\pi} \mathcal{P} \int_{0}^{1} \frac{\tilde{\theta}(s')-\pi}{s' - s} ds' ,
    \label{eqn:15}
\end{align}
where the constants are fixed by the requirement that $\tilde{q} \rightarrow 1$ as $\tilde{\phi} \rightarrow -\infty$.

Note that if $\Tilde{S}$ measures (dimensionless) arclength along the interface from the nose, the radius of curvature will take the simple form as
\begin{align}
    \frac{1}{\tilde{R}} = \frac{d\tilde{\theta}}{d\tilde{S}} = \frac{\partial \tilde{\theta}}{\partial \tilde{\phi}}\frac{d\tilde{\phi}}{d\tilde{S}} = \tilde{q}\frac{\partial \tilde{\theta}}{\partial \tilde{\phi}}.
    \label{eqn:16}
\end{align}
Using $s$ as the independent variable and differentiating Eq. \eqref{eqn:14} with respect to s, we have
\begin{align}
    \frac{d\tilde{\phi}}{ds} = -\frac{1}{\pi s}.
    \label{eqn:17}
\end{align}
The other relations that can be used are
\begin{align}
    \frac{d\tilde{x}}{d\tilde{S}} = \cos{\tilde{\theta}},~~~~ \frac{d\tilde{\phi}}{d\tilde{S}} = \tilde{q}.
    \label{eqn:18}
\end{align}
\subsection{Derivation of Integro-Differential equations for the interface} \label{subsec:Integrodiff}
Using the dimensionless parameters \eqref{eqn:12} and the relation \eqref{eqn:16} in the equation of the interface \eqref{eqn:13}
\begin{align}
    \tilde{\phi} = \frac{h_0^3\nu}{12\mu(1-\Lambda)Uw^3}\tilde{h}^3\tilde{q}\frac{\partial \tilde{\theta}}{\partial \tilde{\phi}} - \frac{\tilde{x}}{(1-\Lambda)} - \frac{x_0}{(
     1-\Lambda)w}+\frac{2h_0^2\nu}{12\mu(1-\Lambda)Uw^2}\tilde{h}^2 .
     \label{eqn:19}
\end{align}
Differentiate \eqref{eqn:19} along the interface and use the relations \eqref{eqn:17} and \eqref{eqn:18} to obtain the second relation between $\tilde{q}$ and $\tilde{\theta}$
\begin{align}
    \left\{ \frac{h_0^3\nu \pi^2}{12\mu Uw^3}\right\}s\tilde{q}\frac{d}{ds}\left[\tilde{h}^3\tilde{q}s\frac{d\tilde{\theta}}{ds} - \frac{2w}{h_0\pi}\tilde{h}^2\right] - \tilde{q}(1-\Lambda) = \cos{\tilde{\theta}} ~; ~~ 0<s<1~.
    \label{eqn:20}
\end{align}
The appropriate boundary conditions for \eqref{eqn:15} and \eqref{eqn:20} are
\begin{gather}
\left.\begin{aligned}
   \tilde{\theta}(0) = \pi, ~~~~ \tilde{q}(0) = \frac{1}{(1-\Lambda)}, ~~\\ \tilde{\theta}(1) = \frac{\pi}{2}, ~~~~ \tilde{q}(1) = 0.~~
    \end{aligned}
    \right\}
    \label{eqn:21}
\end{gather}
From Eqs. \eqref{eqn:15} and \eqref{eqn:21}, we obtain
\begin{align}
    \ln(1-\Lambda) = \frac{1}{\pi} \mathcal{P} \int_{0}^{1} \frac{\tilde{\theta}-\pi}{s'}ds'. 
    \label{eqn:22}
\end{align}
Set
\begin{align}
    \theta = \tilde{\theta} - \pi , \quad \tilde{q} = \frac{q}{(1-\Lambda)} ~\text{and}~ Ca_m \equiv \frac{\pi^2}{(1-\Lambda)^2}\frac{h_0^3}{w^3}\frac{1}{Ca} ~\text{where,}~ Ca = \frac{12\mu U}{\nu}
    \label{eqn:23}
\end{align}
and rewrite Eq. \eqref{eqn:15} and \eqref{eqn:20} as
\begin{align}
    \ln{q} = -\frac{s}{\pi} \mathcal{P} \int_{0}^{1} \frac{\theta(s')}{s'(s'-s)}ds', \label{eqn:24}\\
    Ca_m sq \frac{d}{ds}\left[ \tilde{h}^3sq\frac{d\theta}{ds}-\frac{2w(1-\Lambda)}{h_0\pi}\tilde{h}^2 \right] = q - \cos{\theta},
    \label{eqn:25}
\end{align}
along with the boundary condition
\begin{align}
    \left.\begin{aligned}
    \theta(0) = 0 , ~~~~ q(0) = 1,~~\\
    \theta(1) = -\frac{\pi}{2}, ~~~~ q(1) = 0.~~
    \end{aligned}
    \right\}
    \label{eqn:26}
\end{align}
Here, $Ca$ and $Ca_m$ are the Capillary number and the modified Capillary number, respectively.
A crucial feature of the integro-differential equations \eqref{eqn:24} and \eqref{eqn:25} is the fact that $Ca_m$ appears in the equation as a singular perturbation parameter. 
In the absence of the surface tension, i.e., for $Ca_m = 0$, Eqs. \eqref{eqn:24} Eq. \eqref{eqn:25} can be solved explicitly to yield
\begin{align}
    q_0(s) = \cos{\theta_0(s)} = \left[\frac{1-s}{1+\alpha_1s}\right]^{1/2} ,
    \label{eqn:27}
\end{align}
where, $\displaystyle \alpha_1 = \frac{2\Lambda -1}{(1-\Lambda)^2}$. In the Cartesian coordinates $(x,y)$, this solution reads  
\begin{align}
    y = \frac{2(1-\Lambda)W}{\pi} \ln \cos{\left[\frac{\pi x}{2\Lambda W}\right]}.
     \label{eqn:28}
\end{align}
Note that Eq. \eqref{eqn:27} is automatically consistent with the boundary conditions for any $0 \leq \Lambda \leq 1$, and thus this solution fails to determine the unique selected $\Lambda$. It is this investigation in which the solvability mechanism for pattern selection is used \cite{hong1987pattern}.



\section{Solvability theory}\label{sec:Solvability}

Considering the modified Capillary parameter as a singular perturbation parameter to the problem, the integro-differential equation \eqref{eqn:24}-\eqref{eqn:25} has been linearized around the family of solutions obtained for zero surface tension ($Ca_m = 0$). 
Thus, we obtained an inhomogeneous integro-differential equation for finite-surface-tension corrections to the shape of the Saffman-Taylor finger. 
A necessary condition for the existence of solutions has been obtained via WKB methods. This allows for the computation of the finger width as an increasing function of the modified Capillary parameter and the gap-gradient parameter. 

The main result of the present study on finger selection in a tapered Hele-Shaw cell is summarized in the following proposition. 

\begin{Proposition}
In the limit of $Ca_m \rightarrow 0$, $(q,\theta)$ satisfying Eqs. \eqref{eqn:24} and \eqref{eqn:25}, along with the boundary conditions \eqref{eqn:26} has a unique solution provided $\Lambda$ satisfies the following relation (accurate up to the leading order in $Ca_m$):
\begin{align}
    \Lambda \sim \frac{1}{2} + \left(\frac{Ca_m}{1-3\alpha\frac{x_0}{h_0}}\right)^{2/3}, \qquad\text{for} \quad \Lambda > 1/2.
     \label{eqn:29}
\end{align}
\end{Proposition}

\begin{proof}
We begin the proof of the above statement by first stating the following vital result, which gives us the necessary condition but not a sufficient condition for the existence of a unique solution. 
\begin{Lemma}
    Consider $\Theta \in C^{\infty}(\mathbb{R})$ and a differential operator $\mathcal{L}$ such that
    \begin{align}
    \mathcal{L}\Theta = \bar{R},
    \label{eqn:l1}
    \end{align}
where $\bar{R} \in L^1(\mathbb{R})$. Let $\Theta_0$ be the null eigenvector of the adjoint of $\mathcal{L}$, or $\Theta_0 \in \mathcal{N}(\mathcal{L}^{\dagger})$. Further, define the cusp function, $\mathcal{C}\in L^1(\mathbb{R})$, as follows
\begin{align}
    \mathcal{C} = \int_{-\infty}^{\infty} d\eta \Theta_0 \bar{R}(\eta).
    \label{eqn:d1}
\end{align}
Then $\mathcal{C} \equiv 0$ if $\Theta$ exists uniquely.
\end{Lemma}

\begin{proof}
The proof of the lemma follows from the observation that if we can find a null eigenvector of the adjoint of $\mathcal{L}$, that is, $\mathcal{L}^{+}\Theta_0 = 0$, then we can write
\begin{align}
    \int_{0}^{\infty} d\eta \Theta_0 \mathcal{L}\Theta = \int_{0}^{\infty} d\eta \Theta_0 \bar{R} = \int_{0}^{\infty} d\eta \Theta (\mathcal{L^+}\Theta_0) =0
    \label{eqn:30}
\end{align}
and therefore, the cusp function (defined in \eqref{eqn:d1}) vanishes, i.e., $\mathcal{C} \equiv 0$.

Note that it is the vanishing of the cusp function that gives the condition for the selection of the finger. Rephrasing the idea, we could say that if we can find a $\Theta_0$ such that there exists a sufficiently well-behaved $\Theta$ for which all integrals in Eq. \eqref{eqn:30} are convergent, the cusp function must vanish. However, the vanishing of the cusp function in no way guarantees the existence of $\Theta$.  
\end{proof}

Next, we linearize ~$(q,\theta)$~ around the zero-modified Capillary parameter (or zero-surface tension solution) by defining $\theta_1$ through
\begin{align}
\theta(s) \approx \theta_0(s) + Ca_m \theta_1(s),
\label{eqn:31}
\end{align}
where $\theta_1$ should satisfy the boundary conditions given by $\theta_1(0) =\theta_1(1) =0$. By defining
\begin{align}
\ln{q_1(s)} = -\frac{s}{\pi} \mathcal{P} \int_{0}^{1} \frac{\theta_1(s')}{s'(s'-s)}ds',
\label{eqn:32}
\end{align}
we can write 
\begin{align}
    q \approx q_0(1+Ca_m \ln q_1). 
    \label{eqn:33}
\end{align}
Further, using Eq. \eqref{eqn:27}, \eqref{eqn:31}, we have 
\begin{align}
    \cos{\theta} \approx q_0 - Ca_m \theta_1 \sin{\theta_0}.
    \label{eqn:34}
\end{align}
Substituting linearized version of $(q, \theta)$ in Eq. \eqref{eqn:25} one obtains
\begin{gather}
    Ca_m q_0(1+Ca_m \ln q_1)s\frac{d}{ds}\left[ \tilde{h}^3s q_0(1+Ca_m\ln q_1)\frac{d}{ds}(\theta_0+Ca_m\theta_1) - \frac{2w(1-\Lambda)}{h_0 \pi}\tilde{h}^2 \right] \nonumber\\= q_0(1+Ca_m\ln q_1) -q_0 + Ca_m\theta_1 \sin\theta_0.
         \label{eqn:35}
\end{gather}
A singular perturbation expansion is carried out in which terms quadratic in $Ca_m$ and terms cubic in $Ca_m$ multiplying derivatives of $\theta_1$ are neglected in the above calculation. 
\begin{gather}
   Ca_m \frac{d^2\theta_1}{ds^2}+Ca_m \left\{\frac{1}{q_0 s\tilde{h}^3}\frac{d}{ds}(q_0 s\tilde{h}^3)\right\}\frac{d\theta_1}{ds} + \left\{\frac{-\sin \theta_0}{(q_0 s)^2\tilde{h}^3} \right\}\theta_1 +  \frac{1}{q_0 s^2\tilde{h}^3\pi} \mathcal{P} \int_{0}^{1} ds' \frac{s\theta_1(s')}{s'(s'-s)} \nonumber\\= -\frac{1}{q_0 s\tilde{h}^3}\left\{ \frac{d}{ds}\left[ \tilde{h}^3q_0 s \frac{d\theta_0}{ds} - \frac{2w(1-\Lambda)}{h_0 \pi}\tilde{h}^2\right] \right\}. 
    \label{eqn:36}
\end{gather}
We rewrite the above equation as 
\begin{align}
    Ca_m \frac{d^2\theta_1}{ds^2} + Ca_m P(s) \frac{d\theta_1}{ds} + Q(s)\theta_1 + H(s) \mathcal{P} \int_{0}^{1} ds' \frac{s\theta_1(s')}{s'(s'-s)} = R(s).
    \label{eqn:37}
\end{align}
where,
\begin{align}
\left.\begin{aligned}
    P(s) = \frac{1}{q_0 s\tilde{h}^3}\frac{d}{ds}(q_0 s\tilde{h}^3),~~~~\\
    Q(s) = \frac{-\sin \theta_0}{(q_0 s)^2\tilde{h}^3},~~~~\\
    H(s) = \frac{1}{q_0 s^2\tilde{h}^3\pi},~~~~\\
    R(s) = -\frac{1}{q_0 s\tilde{h}^3}\left\{ \frac{d}{ds}\left[ \tilde{h}^3q_0 s \frac{d\theta_0}{ds} - \frac{2w(1-\Lambda)}{h_0 \pi}\tilde{h}^2\right] \right\}.~~~~
    \end{aligned}
    \right\}
    \label{eqn:38}
\end{align}
Al-Housseiny et al. \cite{al2012control} performed experiment to investigate the effects of gap-gradient in a tapered Hele-Shaw cell supported by linear stability theory when $\alpha^2 \ll 1$ and $\displaystyle \left \lvert \frac{\alpha w}{h_0} \right \rvert \ll 1$. In the current study, under these assumptions we have $\displaystyle \frac{d\tilde{h}}{ds} = 0$. 
Noting that $\sin{\theta_0} = - (1-\cos^2\theta_0)^{1/2}$ in the fourth quadrant, and using equation \eqref{eqn:27} in \eqref{eqn:38}, we obtain
\begin{align}
    \left.\begin{aligned}
    P(s) = \frac{1}{s} - \frac{1+\alpha_1}{2(1-s)(1+\alpha_1 s)},~~\\
    Q(s) =\frac{1}{\left(1+\alpha\frac{x_0}{h_0} \right)^3}\frac{\sqrt{(1+\alpha_1 s)(1+\alpha_1)}}{s^{\frac{3}{2}}(1-s)},~~\\
    H(s) = \frac{1}{\pi}\frac{1}{\left(1+\alpha\frac{x_0}{h_0} \right)^3}\frac{\sqrt{1+\alpha_1 s}}{s^2(\sqrt{1-s})},~~ \\
    R(s) = \frac{\sqrt{1+\alpha_1}}{2\sqrt{s}(1+\alpha_1 s)(\sqrt{1-s})}\left[\frac{1}{2s} -\frac{3\alpha_1}{2(1+\alpha_1 s)} \right].~~
    \end{aligned}
    \right\}
    \label{eqn:39}
\end{align}
Applying the following change of variable
\begin{align}
    \eta = \left[ \frac{1-s}{(1+\alpha_1)s} \right]^{1/2},
    \label{eqn:40}
\end{align}
Eq. \eqref{eqn:37} can be expressed as
\begin{align}
    Ca_m \frac{d^2\theta_1(\eta)}{d\eta^2}+Ca_m \left[\frac{d^2\eta}{ds^2}+P\frac{d\eta}{ds}\right]\frac{1}{(\frac{d\eta}{ds})^2} \frac{d\theta_1 (\eta)}{d\eta}+\frac{Q}{(\frac{d\eta}{ds})^2}\theta_1(\eta)+\frac{H}{(\frac{d\eta}{ds})^2}I = \frac{R}{(\frac{d\eta}{ds})^2}, 
    \label{eqn:41}
\end{align}
where $I$ in Eq. \eqref{eqn:41} represents the integral
\begin{align}
     I = \mathcal{P} \int_{0}^{1} ds' \frac{s\theta_1(s')}{s'(s'-s)}.
     \label{eqn:42}
\end{align}
appearing in Eq. \eqref{eqn:37}. 
It is worth mentioning that the slope $\eta$ varies from $-\infty$ to $\infty$ as one goes all the way around the finger passing through the tip at $\eta = 0$.
Furthermore, Eq. \eqref{eqn:40} can be  recast as $s = 1/(1+\beta^2\eta^2)$, where, $\beta^2 = 1+\alpha_1$.

Next, we employ a change of variable to eliminate the first derivative from Eq. \eqref{eqn:41}. This is done by introducing $\Theta(\eta)$ through $\theta_1(\eta) = g(\eta)\Theta(\eta)$ and $g(\eta)$ is chosen in such a way that the first derivative vanishes. In terms of this new variable, Eq. \eqref{eqn:37} reads
\begin{gather}
    Ca_m\frac{d^2\Theta(\eta)}{d\eta^2}+Ca_m \left[ \frac{\frac{d^2\eta}{ds^2}+P\frac{d\eta}{ds}}{(\frac{d\eta}{ds})^2} + 2\frac{g'}{g} \right] \frac{d\Theta(\eta)}{d\eta} + Ca_m \left[ \frac{g''}{g} +  \frac{\frac{d^2\eta}{ds^2}+P\frac{d\eta}{ds}}{(\frac{d\eta}{ds})^2} \frac{g'}{g}\right]\Theta(\eta) \nonumber \\ ~~~~~~~~~~~~~~~~~~~~~~~~~~~~~~~~+ \frac{Q(\eta)}{(\frac{d\eta}{ds})^2}\Theta(\eta) + \frac{H}{g(\frac{d\eta}{ds})^2}I = \frac{R}{g(\frac{d\eta}{ds})^2}.
    \label{eqn:43}
\end{gather} 
To eliminate the first derivative, we equate the coefficient of $\frac{d\Theta}{d\eta}$ to zero, yielding
\begin{align}
    \ln g = \ln k -\frac{1}{2}\int \frac{\frac{d^2\eta}{ds^2}+P\frac{d\eta}{ds}}{(\frac{d\eta}{ds})^2} d\eta.
    \label{eqn:44}
\end{align}
To proceed further, following the analysis of Corvera \cite{corvera1995anisotropic} and Bansal {\emph et al.} \cite{bansal2023selection}, only those terms that are consistent with our first approximation have been kept; that is, non-singular terms of order $Ca_m$ have been ignored. 
Thus, Eq. \eqref{eqn:43} simplies to 
\begin{align}
     Ca_m\frac{d^2\Theta(\eta)}{d\eta^2}+\bar{Q}_1\Theta + \bar{H} =\bar{R},
     \label{eqn:45}
\end{align}
where,
\begin{align}
    \bar{Q}_1 = \frac{Q(\eta)}{(\frac{d\eta}{ds})^2}, \quad
    \bar{H} = \frac{H}{g(\frac{d\eta}{ds})^2}I, \quad
    \bar{R} = \frac{R}{g(\frac{d\eta}{ds})^2}.
    \label{eqn:46}
\end{align}
Considering that the derivatives of $\eta$ are given by
\begin{align}
    \left.\begin{aligned}
    \frac{d\eta}{ds} = \frac{d}{ds}\left[\frac{1-s}{s(1+\alpha_1)}\right]^{1/2} = -\frac{1}{2\sqrt{(1+\alpha_1)(1-s)}s^{3/2}}, ~~\text{and}~~\\
    \frac{d^2\eta}{ds^2} = -\frac{1}{2\sqrt{(1+\alpha_1)(1-s)}s^{3/2}}\left[\frac{1}{2(1-s)} -\frac{3}{2s}\right],
    \end{aligned}
    \right\}
    \label{eqn:47}
\end{align}
the following from  of $g(s)$ is obtained
\begin{align}
    g(s) = \frac{s^{1/4}(1+\alpha_1 s)^{1/4}}{(1+\alpha_1)^{1/4}},
    \label{eqn:48}
\end{align}
which in terms of new variable $\eta$ is given by
\begin{align}
    g(\eta) = \frac{(1+\eta^2)^{1/4}}{(1+\eta^2 \beta^2)^{1/2}}. 
    \label{eqn:49}
\end{align}
Therefore, 
\begin{align}
    \Theta(\eta) = \frac{\theta_1(\eta)}{g(\eta)} = \frac{(1+\eta^2 \beta^2)^{1/2}}{(1+\eta^2)^{1/4}}\theta_1(\eta).
    \label{eqn:50}
\end{align}
Next, the coefficients of Eq. \eqref{eqn:45} given in Eqs. \eqref{eqn:46} are expressed in terms of the transformed variable $\eta$. We have 
\begin{align}
    I = \mathcal{P} \int_0^1 \frac{s\theta_1(s')}{s'(s'-s)}ds'=\int_{\infty}^0 \frac{(1+\eta'^2)^{1/4}}{(1+\beta^2 \eta'^2)^{1/2}} \left[\frac{1}{\eta+\eta'} -\frac{1}{\eta -\eta'} \right]\Theta(\eta')d\eta'.
    \label{eqn:51}
\end{align}
This leads to 
\begin{align}
    \bar{H} = \frac{1}{\pi(1+\alpha \frac{x_0}{h_0})^3}\int_{\infty}^0 \frac{4\beta^4\eta(1+\eta^2)^{1/4}}{(1+\beta^2\eta^2)^{3/2}} \frac{(1+\eta'^2)^{1/4}}{(1+\beta^2 \eta'^2)^{1/2}} \left[\frac{1}{\eta+\eta'} -\frac{1}{\eta -\eta'} \right]\Theta(\eta')d\eta'.
    \label{eqn:52}
\end{align}
Assuming that $\Theta$ is an odd function, i.e., $\Theta(-\eta') =-\Theta(\eta')$, $\bar{H}$ can be expressed as 
\begin{align}
    \bar{H} = \frac{1}{\pi(1+\alpha \frac{x_0}{h_0})^3}\int_{-\infty}^{\infty} \frac{4\beta^4\eta(1+\eta^2)^{1/4}}{(1+\beta^2\eta^2)^{3/2}} \frac{(1+\eta'^2)^{1/4}}{(1+\beta^2 \eta'^2)^{1/2}}\frac{\Theta(\eta')}{\eta -\eta'}d\eta'.
    \label{eqn:53}
\end{align}
The extended range of $\eta \in (-\infty, \infty)$ including the negative real axis is applicable since only symmetric fingers are considered here. Hence, $\theta_1$ and therefore $\Theta$ are anti-symmetric functions of their arguments. The coefficients of the linear term and the inhomogeneous term are redefined as $\displaystyle \bar{Q}_1 = Q/\left( \frac{d\eta}{ds} \right)^2$ and $ \displaystyle \bar{R} = R/g\left( \frac{d\eta}{ds} \right)^2$. 

Equations \eqref{eqn:45} and \eqref{eqn:46} have been transformed into
\begin{align}
    Ca_m\frac{d^2\Theta(\eta)}{d\eta^2} + \bar{Q}_1\Theta(\eta) + \frac{1}{\pi}\mathcal{P}\int_{-\infty}^{\infty} \frac{\bar{Q}_2(\eta,\eta')\Theta(\eta')}{\eta-\eta'}d\eta' = \bar{R}(\eta),
     \label{eqn:54}
\end{align}
where,
\begin{align}
    \left.\begin{aligned}
    \bar{Q}_1(\eta) = \frac{4\beta^4}{(1+\alpha \frac{x_0}{h_0})^3}\frac{(1+\eta^2)^{1/2}}{(1+\eta^2\beta^2)^2},~~\\
    \bar{Q}_2(\eta,\eta') = \frac{4\beta^4}{(1+\alpha \frac{x_0}{h_0})^3} \frac{\eta(1+\eta^2)^{1/4}}{(1+\beta^2\eta^2)^{3/2}} \frac{(1+\eta'^2)^{1/4}}{(1+\beta^2 \eta'^2)^{1/2}},~~\\
    \bar{R}(\eta) = \frac{\eta[3+\beta^2(\eta^2-2)]}{(1+\beta^2\eta^2)^{1/2}(1+\eta^2)^{9/4}}.~~
    \end{aligned}
    \right\}
    \label{eqn:55}
\end{align}
Equation \eqref{eqn:54} has the form
\begin{align}
    \mathcal{L}\Theta(\eta) = \bar{R}(\eta).
    \label{eqn:56}
\end{align}
Next, we use cusp function definted in \eqref{eqn:d1} to find the equation for the null eigenvector of the adjoint operator $\mathcal{L}^{\dagger}$. Multiplying $\Theta_0$ and integrating the left hand side of Eq. \eqref{eqn:56} can rewritten as
\begin{align}
    \int_{0}^{\infty} d\eta \Theta_0(\eta) \mathcal{L}\Theta(\eta) = \underbrace{\int_{0}^{\infty} d\eta \Theta_0(\eta) Ca_m \frac{d^2\Theta}{d\eta^2}}_{(A)} + \underbrace{\int_{0}^{\infty} d\eta \Theta_0(\eta)\bar{Q}_1\Theta(\eta)}_{(B)} ~~~~~~~~~~~~~~~\nonumber \\+ \int_{0}^{\infty} d\eta \Theta_0(\eta) \frac{1}{\pi} \left\{\underbrace{\mathcal{P}\int_{-\infty}^{0} \frac{\bar{Q}_2(\eta,\eta')\Theta(\eta')}{\eta-\eta'}d\eta'}_{(C)} + \underbrace{ \mathcal{P}\int_{0}^{\infty} \frac{\bar{Q}_2(\eta,\eta')\Theta(\eta')}{\eta-\eta'}d\eta'}_{(D)} \right\}. 
    \label{eqn:57}
\end{align}
Integrating by parts, the integral $(A)$ in \eqref{eqn:57} can be written as, 
\begin{gather}
    \int_{0}^{\infty} d\eta \Theta_0(\eta) Ca_m \frac{d^2\Theta}{d\eta^2} = Ca_m \int_{0}^{\infty} d\eta \Theta_0(\eta) \frac{d}{d\eta}\left(\frac{d\Theta}{d\eta}\right) = Ca_m \left[ (\Theta_0\Theta' - \Theta_0'\Theta)|_{0}^{\infty}\right ] \nonumber \\ + \int_{0}^{\infty} Ca_m \Theta \frac{d^2\Theta_0}{d\eta^2} d\eta= \int_{0}^{\infty} Ca_m \Theta \frac{d^2\Theta_0}{d\eta^2} d\eta.
    \label{eqn:58}
\end{gather}
Analysis of the term in the square bracket in \eqref{eqn:58}, yields that this term vanishes on the side of the finger, for which $\eta = \infty$ ($s = 0$, see \eqref{eqn:40}) \cite{corvera1995saffman, hong1986analytic}. Next, we consider this term at the fingertip ($\eta = 0$). Although the explicit form of $\Theta$ is unknown, using relation \eqref{eqn:50}, $\Theta$ can be expressed in terms of the angle at the tip of the finger. From \eqref{eqn:50}, we find that when $\eta = 0$, i.e., at the tip of the finger, $\left. \Theta \right \rvert_{tip} = \left. \theta_1 \right \rvert_{tip}$ and $\displaystyle \left. \frac{d\theta_1}{d\eta} \right\rvert_{tip} = \displaystyle \left. \frac{d\theta_1}{d\eta} \right\rvert_{tip}$. Furthermore, we find that $\displaystyle \left. \frac{d \Theta}{d\eta} \right\rvert_{tip} = \left( \left. \frac{d\theta_1}{ds} \right \rvert_{tip} \right) \left. \left( \frac{ds}{d\eta} \right\rvert_{tip} \right) = 0$ as $\displaystyle \left. \frac{ds}{d\eta} \right\rvert_{tip} = 0$. Therefore, $\left. \Theta_0 \Theta^\prime \right \rvert_{\eta = 0}$ in \eqref{eqn:58} vanishes. For physically meaningful solutions, we require $\theta_1|_{tip}$ to vanish (or, $\theta_1(0) = 0$) and thus $\left. \Theta_0^\prime \Theta \right \rvert_{\eta = 0}$ also vanishes. 

Term $(B)$ in \eqref{eqn:57} can be rewritten as
\begin{align}
    \int_{0}^{\infty} d\eta \Theta_0(\eta)\bar{Q}_1\Theta(\eta) = \int_{0}^{\infty} d\eta \Theta(\eta)\bar{Q}_1\Theta_0(\eta).
    \label{eqn:59}
\end{align}
By changing $\eta' \rightarrow -\eta'$ and $\eta \rightarrow -\eta$, term $(C)$ in \eqref{eqn:57} becomes
\begin{align*}
    \int_{0}^{\infty} d\eta \Theta_0(\eta) \frac{1}{\pi}\mathcal{P}\int_{-\infty}^{0} \frac{\bar{Q}_2(\eta,\eta')\Theta(\eta')}{\eta-\eta'}d\eta'=  \int_{-\infty}^{0} d\eta' \Theta(\eta') \frac{1}{\pi}\mathcal{P}\int_{0}^{\infty} \frac{\bar{Q}_2(\eta,\eta')\Theta_0(\eta)}{\eta-\eta'}d\eta\\ = \int_{\infty}^{0} -d\eta' \Theta(-\eta') \frac{1}{\pi}\mathcal{P}\int_{0}^{-\infty} -d\eta\frac{\bar{Q}_2(-\eta,-\eta')\Theta_0(-\eta)}{\eta'-\eta}.
\end{align*}
Noting $\bar{Q}_2(-\eta, -\eta') = -\bar{Q}_2(\eta, \eta')$ and $\Theta(-\eta') = -\Theta(\eta')$, and the fact that $\Theta_0(\eta)$ was obtained such that $\Theta_0(-\eta) = -\Theta_0(\eta)$, term $(C)$ reduces to
\begin{gather}
    \int_{\infty}^{0} d\eta' \Theta(\eta') \frac{1}{\pi}\mathcal{P}\int_{0}^{-\infty}d\eta\frac{\bar{Q}_2(\eta,\eta')\Theta_0(\eta)}{\eta-\eta'}
    = \int_{0}^{\infty} d\eta \Theta(\eta) \frac{1}{\pi}\mathcal{P}\int_{-\infty}^{0}d\eta'\frac{\bar{Q}_2(\eta',\eta)\Theta_0(\eta')}{\eta'-\eta},
    \label{eqn:60}
\end{gather}
where we have changed $\eta' \rightarrow \eta$ and $\eta \rightarrow \eta'$ to write the last equality. By changing $\eta' \rightarrow \eta$ and $\eta \rightarrow \eta'$, we can write the fourth term, $(D$, in Eq. \eqref{eqn:57} as
\begin{align}
    \int_{0}^{\infty} d\eta \Theta_0(\eta) \frac{1}{\pi} \mathcal{P}\int_{0}^{\infty} \frac{\bar{Q}_2(\eta,\eta')\Theta(\eta')}{\eta-\eta'}d\eta' = \int_{0}^{\infty} d\eta \Theta(\eta) \frac{1}{\pi} \mathcal{P}\int_{0}^{\infty} \frac{\bar{Q}_2(\eta',\eta)\Theta_0(\eta')}{\eta'-\eta}d\eta'.
    \label{eqn:61}
\end{align}
Combining Eqs. \eqref{eqn:58}--\eqref{eqn:61} in Eq. \eqref{eqn:57}, we rewrite 
\begin{align}
    \int_{0}^{\infty} d\eta \Theta_0(\eta) \mathcal{L}\Theta(\eta) = Ca_m \int_{0}^{\infty} d\eta\Theta \frac{d^2\Theta_0}{d\eta^2} + \int_{0}^{\infty} d\eta \Theta(\eta)\bar{Q}_1\Theta_0(\eta) \nonumber\\+ \int_{0}^{\infty} d\eta \Theta(\eta) \frac{1}{\pi} \mathcal{P}\int_{-\infty}^{\infty} \frac{\bar{Q}_2(\eta',\eta)\Theta_0(\eta')}{\eta'-\eta}d\eta'.
    \label{eqn:62}
\end{align}
We define the operator $\mathcal{L}^{\dagger}$ such that
\begin{align*}
    \mathcal{L}^{\dagger} \Theta_0 \equiv Ca_m \frac{d^2\Theta_0}{d\eta^2} + \bar{Q}_1\Theta_0 + \frac{1}{\pi} \mathcal{P}\int_{-\infty}^{\infty} \frac{\bar{Q}_2(\eta',\eta)\Theta_0(\eta')}{\eta'-\eta}.
\end{align*}
Now, as stated earlier, if we can find a solution to the equation 
\begin{align}
    \mathcal{L^{\dagger}}\Theta_0 = 0, ~~~\text{with}~ \Theta_0(-\eta)=-\Theta_0(\eta), 
    \label{eqn:63}
\end{align} 
the right hand side of Eq. \eqref{eqn:62} vanishes and the cusp function $\mathcal{C}=2\int_{0}^{\infty} d\eta \Theta_0(\eta) \mathcal{L}\Theta(\eta)=2\int_{0}^{\infty} d\eta \Theta_0(\eta)\bar{R}$ vanishes as well. 

What is done next is to use the WKB method to compute a null eigenvector of the operator $\mathcal{L}^{\dagger}$ that is, to find a solution to Eq. \eqref{eqn:63}. 

\subsection{WKB Approximation} \label{subsec:WKB}
Suppose that the solutions of Eq. \eqref{eqn:63} have the WKB form $e^{\frac{S}{\sqrt{Ca_m}}}$, where $\Re(S) < 0$, and it has a points of stationary phase (that is points $\bar{\eta}$ where $S'(\bar{\eta})=0$ ). Then, in the limit of $Ca_m \rightarrow 0$, the integral can be evaluated by expanding the exponent around the point of the stationary phase:
\begin{align*}
     \mathcal{P}\int_{-\infty}^{\infty} d\eta' \frac{\bar{Q}_2(\eta',\eta)\Theta_0(\eta')}{\eta'-\eta} \simeq e^{\frac{S(\eta)}{\sqrt{Ca_m}}} \mathcal{P}\int_{-\infty}^{\infty} d\eta' \frac{\bar{Q}_2(\eta',\eta) e^{\frac{S''(\eta)(\eta'-\eta)^2}{2! \sqrt{Ca_m}}} }{\eta'-\eta} ~~~~(\text{since}~ S'(\eta)=0).
\end{align*}
The only contribution to the integral which is not exponentially small comes from the pole at $\eta'=\eta$, hence 
\begin{align*}
    \mathcal{P}\int_{-\infty}^{\infty} d\eta' \frac{\bar{Q}_2(\eta',\eta)\Theta_0(\eta')}{\eta'-\eta} = \pi i e^{\frac{S(\eta)}{\sqrt{Ca_m}}} e^{\frac{S''(\eta)(\eta-\eta)^2}{2! \sqrt{Ca_m}}} \bar{Q}_2(\eta,\eta) = \pi i \bar{Q}_2(\eta,\eta) \Theta_0 ~~~~(\text{as,}~\Theta_0 = e^{\frac{S(\eta)}{\sqrt{Ca_m}}} ).
\end{align*}
Also, define $\bar{Q} = \bar{Q}_1 + i \bar{Q}_2(\eta,\eta)$. 
The equation for the null eigenvector of $\mathcal{L}^{\dagger}$ becomes
\begin{align}
    Ca_m \frac{d^2\Theta_0(\eta)}{d\eta^2} + \bar{Q}\Theta_0(\eta) = 0, 
    \label{eqn:64}\\
    \text{where} \quad \bar{Q} = \frac{1}{(1+\alpha \frac{x_0}{h_0})^3}\frac{4\beta^4 (1+i\eta)^{3/2} (1-i\eta)^{1/2}}{(1+\beta^2\eta^2)^2}.
    \label{eqn:65}
\end{align}
In what follows we find the WKB solution of Eq. \eqref{eqn:64}. 

As mentioned above, we search for a solution of the form $\Theta_0 \sim e^{\frac{S}{\sqrt{Ca_m}}}$, where $S$ can be expanded in powers of $Ca_m^{1/2}$, that is,
\begin{align}
    S = \sum_{n=0}^{\infty} S_n Ca_m^{n/2}.
    \label{eqn:wkb}
\end{align}
Here, we have started with the first two terms, that is, $S_0$ and $S_1$. (In Appendix \ref{sec:re-exam} we have shown that a WKB analysis incorporating all terms in the expansion changes the results approximately 5\% in the pre-factor.) Therefore, we have
\begin{align}
    \left.\begin{aligned}
    \Theta_0 = e^{\frac{S_0}{\sqrt{Ca_m}}+S_1},\\
    \frac{d\Theta_0}{d\eta} = (\frac{S'_0}{\sqrt{Ca_m}}+S'_1)e^{\frac{S_0}{\sqrt{Ca_m}}+S_1},\\
    \frac{d^2\Theta_0}{d\eta^2} = \left\{(\frac{S''_0}{\sqrt{Ca_m}}+S''_1)+(\frac{S'_0}{\sqrt{Ca_m}}+S'_1)^2\right\}e^{\frac{S_0}{\sqrt{Ca_m}}+S_1}.
    \end{aligned}
    \right\}
    \label{eqn:66}
\end{align}
Now substituting Eq. \eqref{eqn:66} in Eq. \eqref{eqn:64}, we have
\begin{align}
    Ca_m \left(\frac{S''_0}{\sqrt{Ca_m}}+S''_1 + \frac{{S'_0}^2}{Ca_m}+2\frac{S_0' S_1'}{\sqrt{Ca_m}} +{S'_1}^2\right) + \bar{Q} = 0.
\end{align}
By equating terms order by order in $Ca_m$, only to orders consistent with the approximation, we have,
\begin{enumerate}
    \item ${S_0'}^2 + \bar{Q} =0 ~~({Ca_m}^0$-th order),
    \item $S_0'' + 2S_0'S_1' =0 ~~({Ca_m}^{1/2}$-th order),
\end{enumerate}
which gives a set of equations for $S_0$ and $S_1$ as
\begin{gather}
    S_0 = i \int_{0}^{\eta} \bar{Q}^{1/2} d\eta ,
    \label{eqn:67}\\
    S_1 = \ln \frac{1}{\bar{Q}^{1/4}}.
    \label{eqn:68}
    \end{gather}
Eq. \eqref{eqn:67} and Eq. \eqref{eqn:68} gives
    \begin{gather}
    \Theta_0 = \frac{e^{S_0/\sqrt{Ca_m}}}{\bar{Q}^{1/4}} , 
    \label{eqn:69}\\
    S_0 = i \int_{0}^{\eta} \bar{Q}^{1/2} d\eta = \frac{2i\beta^2}{(1+\alpha \frac{x_0}{h_0})^{3/2}}\int_{0}^{\eta} \frac{(1+i\eta')^{3/4} (1-i\eta')^{1/4}}{(1+\beta^2\eta'^2)} d\eta'.
    \label{eqn:70}
\end{gather}
Note that the complex conjugate of $\Theta_0$ is also a solution of Eq. \eqref{eqn:63}. So, the appropriate antisymmetric combination of the solution is given by 
\begin{align}
\frac{1}{2i}[\Theta_0 - \Theta^*_0] = \Im\Theta_0.
\end{align}
Therefore, the cusp function can be written as
\begin{align}
    \mathcal{C}(\Lambda, Ca_m,\alpha) = \int_{-\infty}^{\infty} d\eta \Im \Theta_0 \bar{R}(\eta) = \int_{-\infty}^{\infty} d\eta \mathcal{F}(\eta) e^{S_0/\sqrt{Ca_m}},
    \label{eqn:71}\\
    \text{where,}~~\mathcal{F} = \frac{\bar{R}}{i\bar{Q}^{1/4}} = -(1+\alpha\frac{x_0}{h_0})^{3/4}\frac{i\eta[3+\beta^2(\eta^2-2)]}{\sqrt{2}\beta(1+\eta^2)^{5/2}}\left(\frac{1-i\eta}{1+i\eta}\right)^{1/8},
    \label{eqn:72}
\end{align}
and $\displaystyle \beta \equiv (1+\alpha_1)^{1/2} = \frac{\Lambda}{1-\Lambda}$. 
In the limit of small $Ca_m, \mathcal{C}(\Lambda, Ca_m, \alpha)$ can be evaluated by the method of steepest descent. We can expect that the behaviour of the integral in Eq. \eqref{eqn:71} will be determined by the behaviour of the integrand at the point of stationary phase.

From Eq. \eqref{eqn:70}, we see that $\eta =i$ is the point of stationary phase (as $S_0^\prime = i\bar{Q}^{1/2} = 0$ implies $\eta = \pm i $), while $\displaystyle \eta = \eta_b = \frac{i}{\beta}$ is a logarithmic branch point of the integral in \eqref{eqn:71} (since $\displaystyle 1 + \beta^2\eta^2 = 0 \Rightarrow \eta = \pm \frac{i}{\beta} $). Hence, there are two qualitatively different situations (as can be seen in Fig. \ref{fig:5}$(a)$ and Fig. \ref{fig:5}$(b)$). 

\begin{figure}[h]
    \centering
    (a) \hspace{2.6 in} (b) \\ 
    \includegraphics[width=\textwidth]{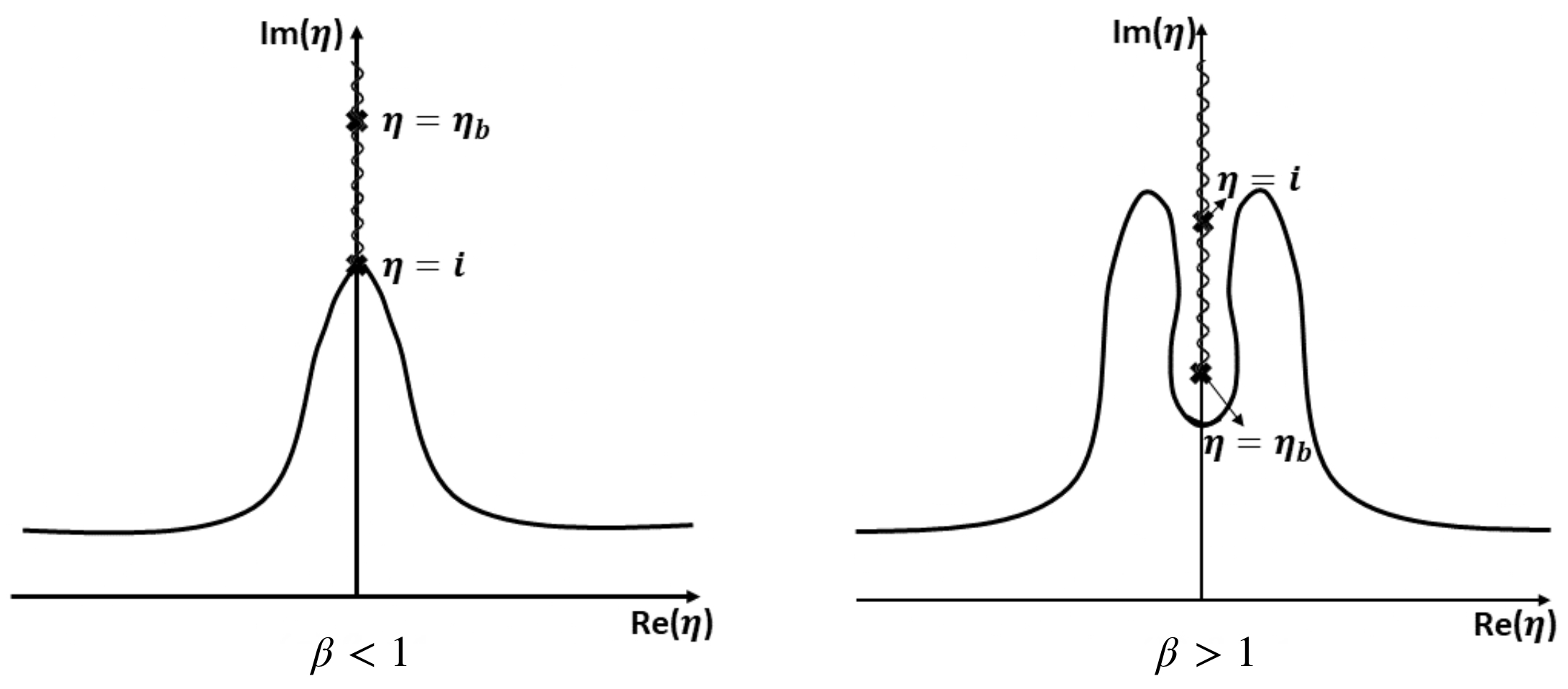}
    \caption{Contours of integration for the evaluation of the cusp function, $\mathcal{C}$ (Definition \eqref{eqn:d1})}
    \label{fig:5}
\end{figure}
\vspace{0.5cm}

\subsubsection{Case-I ($\beta<1$ or $\Lambda < 1/2$)}
Here, the branch point $\displaystyle \eta_b = \frac{i}{\beta}$ is out of the contour of integration for evaluating $\mathcal{C}$. Thus, we can integrate by expanding $S_0(\eta)$ around the point of the stationary phase $\eta =i$. Defining $\omega$ through $\eta= i + \omega$, we write $S_0$ as 
\begin{align}
    S_0 = \frac{2i\beta^2}{(1+\alpha \frac{x_0}{h_0})^{3/2}}\left[ \underbrace{\int_{0}^{i} \frac{(1+i\eta')^{3/4} (1-i\eta')^{1/4}}{(1+\beta^2\eta'^2)} d\eta'}_{(A_1)} + \underbrace{\int_{i}^{i+\omega} \frac{(1+i\eta')^{3/4} (1-i\eta')^{1/4}}{(1+\beta^2\eta'^2)} d\eta'}_{(B_1)}\right].
    \label{eqn:73}
\end{align}
For the term $(A_1)$ in Eq. \eqref{eqn:73} we make the change of variable $u = -i\eta^\prime$ and label it as $E(\Lambda, \alpha)$,  
\begin{align}
    E(\Lambda, \alpha)= -\frac{2\beta^2}{(1+\alpha\frac{x_0}{h_0})^{3/2}} \int_{0}^{1} \frac{(1-u)^{3/4}(1+u)^{1/4}}{1-\beta^2u^2}du, 
    \label{eqn:74}
\end{align}
which is independent of $\omega$. 
For the term $(B_1)$ in Eq. \eqref{eqn:73}, we approximate for small $\omega$, so we can write $S_0$ as 
\begin{align}
    S_0 \simeq E(\Lambda, \alpha) + \frac{2i\beta^2}{(1+\alpha\frac{x_0}{h_0})^{3/2}} \int_{0}^{\omega} \frac{(i\omega')^{3/4}(2-i\omega')^{1/4}}{1-\beta^2}d\omega' \nonumber \\
    \simeq E(\Lambda, \alpha) + \frac{2^{5/4}i^{7/4}}{(1+\alpha\frac{x_0}{h_0})^{3/2}}\frac{\beta^2}{1-\beta^2}\int_{0}^{\omega} (\omega')^{3/4}d\omega'.
    \label{eqn:75}
\end{align}
Therefore, close to the point of the stationary phase, Eq. \eqref{eqn:75} becomes
\begin{align}
    S_0 \simeq E(\Lambda, \alpha) + \frac{8}{7}\frac{2^{1/4}i^{7/4}}{(1+\alpha\frac{x_0}{h_0})^{3/2}}\frac{\Lambda^2}{1-2\Lambda}\omega^{7/4}.
    \label{eqn:76}
\end{align}
For consistency, we also expand $\mathcal{F}$ in the neighbourhood of $\eta = i$ (that is $\eta = i + \omega, |\omega| \ll 1$).
\begin{align}
    F(\omega) \simeq -(1+\alpha\frac{x_0}{h_0})^{3/4} \frac{3}{8}\frac{2^{1/8}}{i^{5/8}}\frac{1-2\Lambda}{\Lambda(1-\Lambda)}\frac{1}{\omega^{21/8}}.
    \label{eqn:77}
\end{align}
In this case, the cusp function is then
\begin{gather*}
    \mathcal{C} = b \int_{-\infty}^{\infty} \omega^{-21/8} e^{\frac{E(\Lambda,\alpha)}{\sqrt{Ca_m}}} e^{\frac{a\omega^{7/4}}{\sqrt{Ca_m}}} d\omega, 
\end{gather*}
which can be expressed in terms of a Gamma function as 
\begin{gather}
     \mathcal{C} = \frac{8\pi i}{7}b \left(\frac{a}{\sqrt{Ca_m}}\right)^{13/14} \frac{1}{\Gamma(\frac{27}{14})}e^{\frac{E(\Lambda,\alpha)}{\sqrt{Ca_m}}},
\end{gather}
where $\displaystyle a = \frac{8}{7}\frac{2^{1/4}i^{7/4}}{(1 + \alpha\frac{x_0}{h_0})^{3/2}}\frac{\Lambda^2}{1 - 2\Lambda}, b = -(1 + \alpha\frac{x_0}{h_0})^{3/4}\frac{3}{8} \frac{2^{1/8}}{i^{5/8}}\frac{1 - 2\Lambda}{\Lambda(1 - \Lambda)}$. Thus, we obtain 
\begin{align}
    \mathcal{C} = \frac{N}{(1+\alpha\frac{x_0}{h_0})^{9/14}} \frac{(1-2\Lambda)^{1/14}\Lambda^{6/7}}{(1-\Lambda)}\frac{1}{Ca_m^{13/28}}e^{\frac{E(\Lambda,\alpha)}{\sqrt{Ca_m}}},~~~~~~~~\text{for}~\Lambda <1/2,
    \label{eqn:78}
\end{align}
where
\begin{align}
    N = 3\pi \frac{2^{22/7}}{7^{27/14}\Gamma(\frac{27}{14})} = 2.008 .
    \label{eqn:79}
\end{align}
Eq. \eqref{eqn:78} illustrates that for a finite $Ca_m$, the cusp function is exponentially small and cannot vanish. The Cusp function, $\mathcal{C}$, has no regular series expansion in powers of $Ca_m$ and $\alpha$, that is, it implies the absence of solutions at non-zero $Ca_m$ for $\Lambda < 1/2$.

For $\alpha = 0$, we recover one of the main results of Hong and Langer \cite{hong1986analytic, hong1987pattern} for the viscous fingering problem. 

\subsubsection{Case-II ($\beta > 1$ or $\Lambda > 1/2)$: Selection}
For $\Lambda > 1/2$, the integral in the exponential term of $\Theta_0$ has a pole, and the contour of integration for $\mathcal{C}$ must go around the branch point $\eta =\bar{\eta}$. The stationary point $\bar{\eta} = i$ must be interpreted as a complex conjugate pair of points on the other side of a branch cut running from $\eta_b = i/\beta$ to $+i\infty$. Accordingly, as shown in the figure \ref{fig:5}(b), the path of steepest descent must include a section which goes from $\eta = i - \delta (\delta \rightarrow 0)$, around the branch point at $i/\beta$ and back up to $i+\delta$. Defining $\eta = i + \omega$ and $\delta$ such that $\delta \rightarrow 0$, we can exprerss $S_0$ as
\begin{gather}
    S_0 = \frac{2i\beta^2}{(1+\alpha \frac{x_0}{h_0})^{3/2}}\left[\lim_{\delta \rightarrow 0} \underbrace{\int_{0}^{i-\delta} \frac{(1+i\eta')^{3/4} (1-i\eta')^{1/4}}{(1+\beta^2\eta'^2)} d\eta'}_{(I_1)} + \lim_{\delta \rightarrow 0}\underbrace{\int_{i-\delta}^{i+\delta} \frac{(1+i\eta')^{3/4} (1-i\eta')^{1/4}}{(1+\beta^2\eta'^2)} d\eta'}_{(I_2)}\right. \nonumber \\+\left. \lim_{\delta \rightarrow 0} \underbrace{\int_{i+\delta}^{i+\omega} \frac{(1+i\eta')^{3/4} (1-i\eta')^{1/4}}{(1+\beta^2\eta'^2)} d\eta'}_{(I_3)}\right].
    \label{eqn:80}
\end{gather}
The term $(I_1)$ in Eq. \eqref{eqn:80} can be expressed as
\begin{gather*}
    \lim_{\delta \rightarrow 0}\frac{2i\beta^2}{(1+\alpha \frac{x_0}{h_0})^{3/2}} \int_{0}^{i-\delta} \frac{(1+i\eta')^{3/4} (1-i\eta')^{1/4}}{(1+\beta^2\eta'^2)} d\eta'\\= \frac{2i\beta^2}{(1+\alpha \frac{x_0}{h_0})^{3/2}}\int_{0}^{i} \frac{(1+i\eta')^{3/4} (1-i\eta')^{1/4}}{(1+\beta^2\eta'^2)} d\eta'.
\end{gather*}
Applying a change of variable $u = -i\eta'$, we have shown that this term is independent of $\omega$, and we denote 
\begin{align}
    E(\Lambda,\alpha) = -\frac{2\beta^2}{(1+\alpha \frac{x_0}{h_0})^{3/2}}\int_{0}^{1} \frac{(1-u)^{3/4} (1+u)^{1/4}}{(1-\beta^2u^2)} du.
    \label{eqn:81}
\end{align}
The term $(I_2)$ in Eq. \eqref{eqn:80} contains a simple pole at $\eta_b = i/\beta$; therefore, Cauchy Residue theorem yields,
\begin{align}
    I_2 =  \lim_{\delta \rightarrow 0}\frac{2i\beta^2}{(1+\alpha \frac{x_0}{h_0})^{3/2}}\int_{i-\delta}^{i+\delta} \frac{(1+i\eta')^{3/4} (1-i\eta')^{1/4}}{(1+\beta^2\eta'^2)} d\eta' \nonumber\\= \frac{2\pi i}{(1+\alpha \frac{x_0}{h_0})^{3/2}} \frac{(2\Lambda -1)^{3/4}}{(1-\Lambda)}.
    \label{eqn:82}
\end{align}

Now, the term $(I_3)$ in Eq. \eqref{eqn:80} becomes
\begin{align}
    I_3 = \lim_{\delta \rightarrow 0}\frac{2i\beta^2}{(1+\alpha \frac{x_0}{h_0})^{3/2}}\int_{i+\delta}^{i+\omega} \frac{(1+i\eta')^{3/4} (1-i\eta')^{1/4}}{(1+\beta^2\eta'^2)} d\eta' \nonumber\\= \frac{2i\beta^2}{(1+\alpha \frac{x_0}{h_0})^{3/2}}\int_{i}^{i+\omega} \frac{(1+i\eta')^{3/4} (1-i\eta')^{1/4}}{(1+\beta^2\eta'^2)} d\eta' \nonumber\\ = \frac{8}{7} \frac{2^{1/4}i^{7/4}}{(1+\alpha \frac{x_0}{h_0})^{3/2}}\frac{\Lambda^2}{1-2\Lambda} \omega^{7/4}.
    \label{eqn:83}
\end{align}
Therefore, close to the point of the stationary phase
\begin{align}
    S_0(\eta) \simeq E(\Lambda,\alpha) + \frac{2\pi i}{(1+\alpha \frac{x_0}{h_0})^{3/2}} \frac{(2\Lambda -1)^{3/4}}{(1-\Lambda)} + \frac{8}{7} \frac{2^{1/4}i^{7/4}}{(1+\alpha \frac{x_0}{h_0})^{3/2}}\frac{\Lambda^2}{1-2\Lambda}\omega^{7/4}.
    \label{eqn:84}
\end{align}
The cusp function then becomes
\begin{gather}
    \mathcal{C} = \int_{-\infty}^{\infty} d\omega \mathcal{F}(\omega) e^{\frac{S_0}{\sqrt{Ca_m}}} = \int_{-\infty}^{\infty} b e^{\frac{E(\Lambda,\alpha)}{\sqrt{Ca_m}}} e^{\frac{2\pi i }{(1 + \alpha \frac{x_0}{h_0})^{3/2}} \frac{(2\Lambda -1)^{3/4}}{(1-\Lambda)}} e^{\frac{a\omega^{7/4}}{\sqrt{Ca_m}}} \omega^{-21/8} d\omega \nonumber \\
    = \frac{8\pi i}{7} b \left(\frac{a}{\sqrt{Ca_m}} \right)^{13/14} \frac{e^{E(\Lambda,\alpha)/\sqrt{Ca_m}}}{\Gamma(\frac{27}{14})} \cos \left[ \frac{2\pi}{(1 + \alpha \frac{x_0}{h_0})^{3/2}} \frac{(2\Lambda -1)^{3/4}}{(1-\Lambda)\sqrt{Ca_m}} \right]. 
\end{gather}
For $\Lambda > 1/2$, we obtain 
\begin{align}
    \mathcal{C} = \frac{N}{(1+\alpha \frac{x_0}{h_0})^{9/14}} \frac{(1-2\Lambda)^{1/14}\Lambda^{6/7}}{(1-\Lambda)}\frac{e^{E(\Lambda,\alpha)/\sqrt{Ca_m}}}{Ca_m^{13/28}} \cos \left[ \frac{2\pi}{(1+\alpha \frac{x_0}{h_0})^{3/2}} \frac{(2\Lambda -1)^{3/4}}{(1-\Lambda)\sqrt{Ca_m}} \right]. 
    \label{eqn:85}
\end{align}
The solvability condition can now be satisfied at each zero of the cosine term, that is,
\begin{align*}
    \frac{2\pi}{\left( 1 + \alpha \frac{x_0}{h_0} \right)^{3/2}} \frac{(2\Lambda -1)^{3/4}}{(1-\Lambda)\sqrt{Ca_m}} = \frac{(2n-1) \pi}{2}, \qquad n \geq 1.
\end{align*}
Among these, only the solution associated with the first zero ($n = 1$) that corresponds to the thinnest finger is linearly stable \cite{corvera1995steady} that yields
\begin{align}
    \Lambda - \frac{1}{2} \simeq \frac{1}{2^{5}} \left( 1 + \alpha \frac{x_0}{h_0} \right)^{2} Ca_m^{2/3} \qquad \text{for} \;\; \Lambda > 1/2. 
    \label{eqn:86}
\end{align}
The form of $\Lambda$ for $\Lambda > 1/2 $ will be
\begin{align}
    \Lambda \simeq \frac{1}{2} + \frac{1}{2^{5}} \left( 1 + 2\alpha \frac{x_0}{h_0} \right) Ca_m^{2/3} \qquad \mbox{as} \quad Ca_m \rightarrow 0. 
    \label{eqn:87}
\end{align}

\end{proof}

It is noteworthy that an alternative form of the finger width 
\begin{align}
    \Lambda \simeq \frac{1}{2} + \frac{1}{2^{5}} \left[ \frac{Ca_m}{(1 - 3 \alpha x_0/h_0)} \right]^{2/3} \qquad \mbox{as} \quad Ca_m \rightarrow 0,
    \label{eqn:alternative_rep}
\end{align}
can be obtained by following the same approach discussed in \S \ref{sec:Solvability} when the term $\left[1/\left(1 + \alpha \frac{x_0}{h_0}\right)\right]^3$ appearing in the expression of $Q$ and $H$ of \eqref{eqn:39} are approximated in the first order of $\alpha$ and $\mathbf{O}(\alpha^2)$ are ignored as $\rvert \alpha \lvert \ll 1$. In the expression of \eqref{eqn:alternative_rep}, we have 
\begin{align}
    \left[\frac{1}{1-3\alpha x_0/h_0}\right]^{2/3} = \left(1-3\alpha x_0/h_0\right)^{-2/3} = 1 + 2\alpha x_0/h_0 + \mathbf{O}(\alpha^2).
\end{align}
Ignoring $\mathbf{O}(\alpha^2)$ terms, one recovers \eqref{eqn:87}, establishing that the two alternative representations are identical, correct up to $\mathbf{O}(\alpha)$. 

\section{Summary and Discussions}\label{sec:discussion}


\begin{figure}[!htbp]
    \centering
    \includegraphics[width=0.5\textwidth]{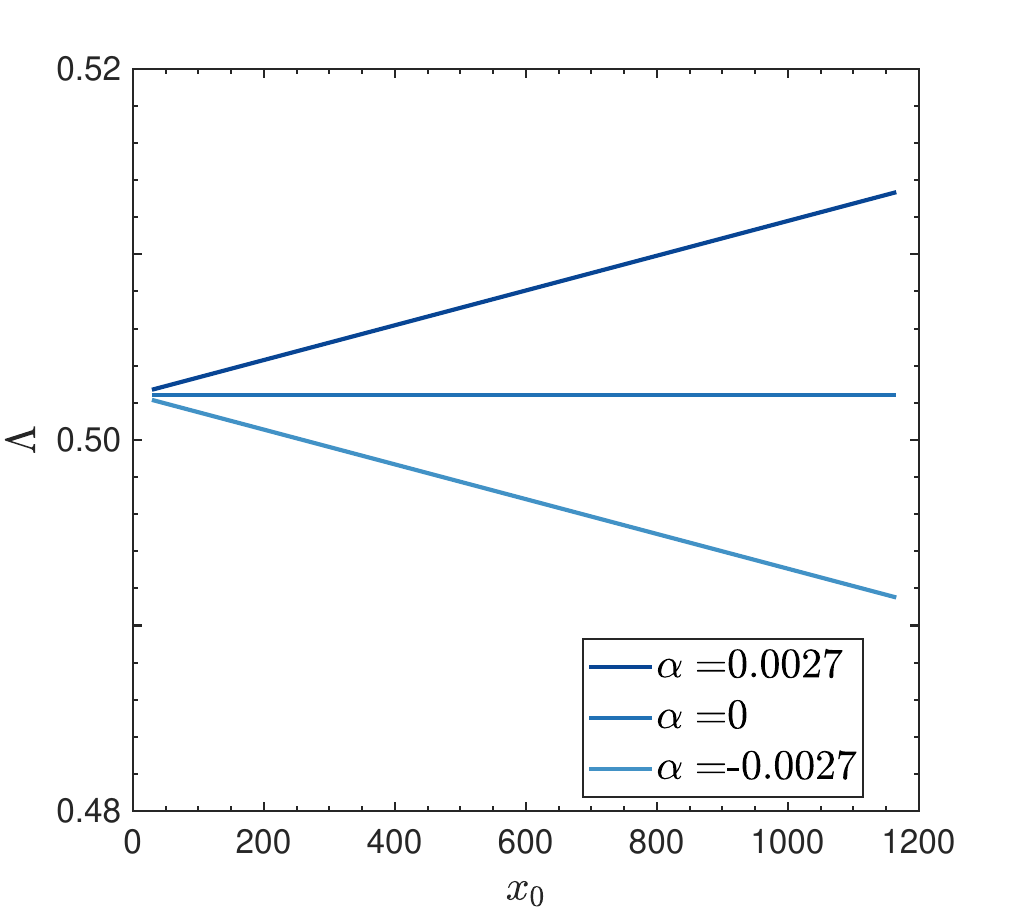}
    \caption{Variation of $\Lambda$ with $x_0$ corresponding to $1/B = 100$. Using Eq. \eqref{eq:B}, we compute $U$ and vary time $t$ to obtain different values of $x_0 = U t$. This shows that for positive gradient ($\alpha = 0.0027$), $\Lambda$ increases, for zero gradient ($\alpha = 0$), it is approximately $1/2$, and for negative gradient ($\alpha = -0.0027$), it decreases.}
    \label{fig:analytic}
\end{figure}

In this paper, we provide an analytical treatment of the Saffman-Taylor fingers in a non-standard Hele-Shaw cell for a class of Newtonian fluids. To find the main result, we have used the solvability condition for the Saffman-Taylor problem in the limit of small modified Capillary number ($Ca_m \rightarrow 0$) and the gap-gradient ($\lvert \alpha \rvert \ll 1$). For $\lvert \alpha \rvert \ll 1$, the physically selected finger is predicted to have a width $\Lambda$ such that $\Lambda - 1/2$ is proportional to 
$\displaystyle \left( 1 + 2\alpha \frac{x_0}{h_0} \right) Ca_m^{2/3}$, 
as $Ca_m \rightarrow 0$. 

In \S \ref{sec:Solvability}, a solvability theory is utilised to derive a necessary condition for pattern selection by performing a singular perturbation expansion around the family of zero-surface-tension solutions. This leads to an inhomogeneous integro-differential equation \eqref{eqn:24}-\eqref{eqn:25} governing the finger shape, in which the operator acts on a function representing the interface orientation angle. The solvability condition requires the inhomogeneous term to be orthogonal to the null space of the adjoint operator, which is enforced through the construction of a cusp function. The zeros of the cusp function determine the admissible finger solutions. As a result, the continuum of solutions is reduced to a discrete set, each corresponding to a finger width that satisfies a power-law relation with the modified capillary parameter. 
Among them, only the solution associated with the first zero of the cosine—corresponding to the thinnest finger—is linearly stable. This is physically reasonable, since wider fingers with smaller tip curvature are prone to the same instabilities that destabilize the flat interface. Thus, solvability theory provides a consistent and physically meaningful mechanism for pattern selection, ensuring a unique finger width for given experimental parameters. From Eq. \eqref{eqn:87}, it is evident that the selected finger width, $\Lambda$, depends not only on the modified capillary number ($Ca_m$) but also on the depth gradient $\alpha$. 



\begin{figure}[!htbp]
     \centering
        (a) \hspace{2.5 in} (b) \\
        \includegraphics[width=0.49\textwidth]{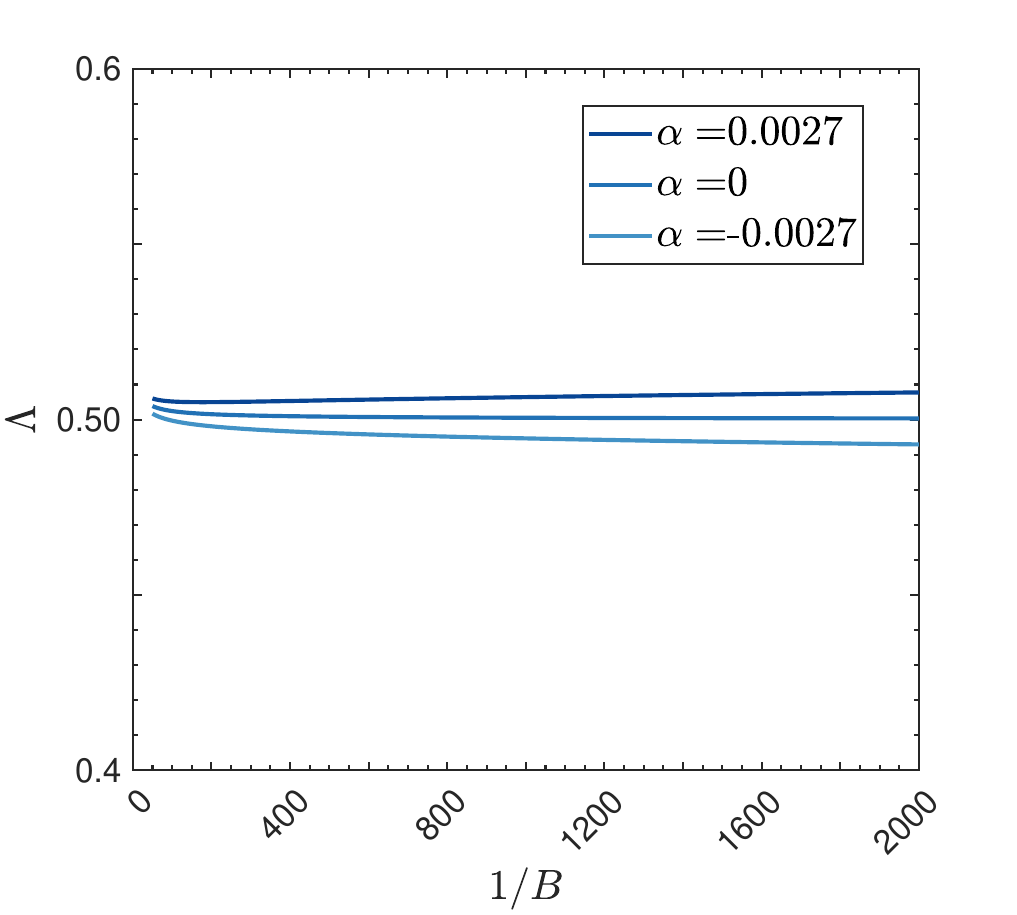}
        \includegraphics[width=0.49\textwidth]{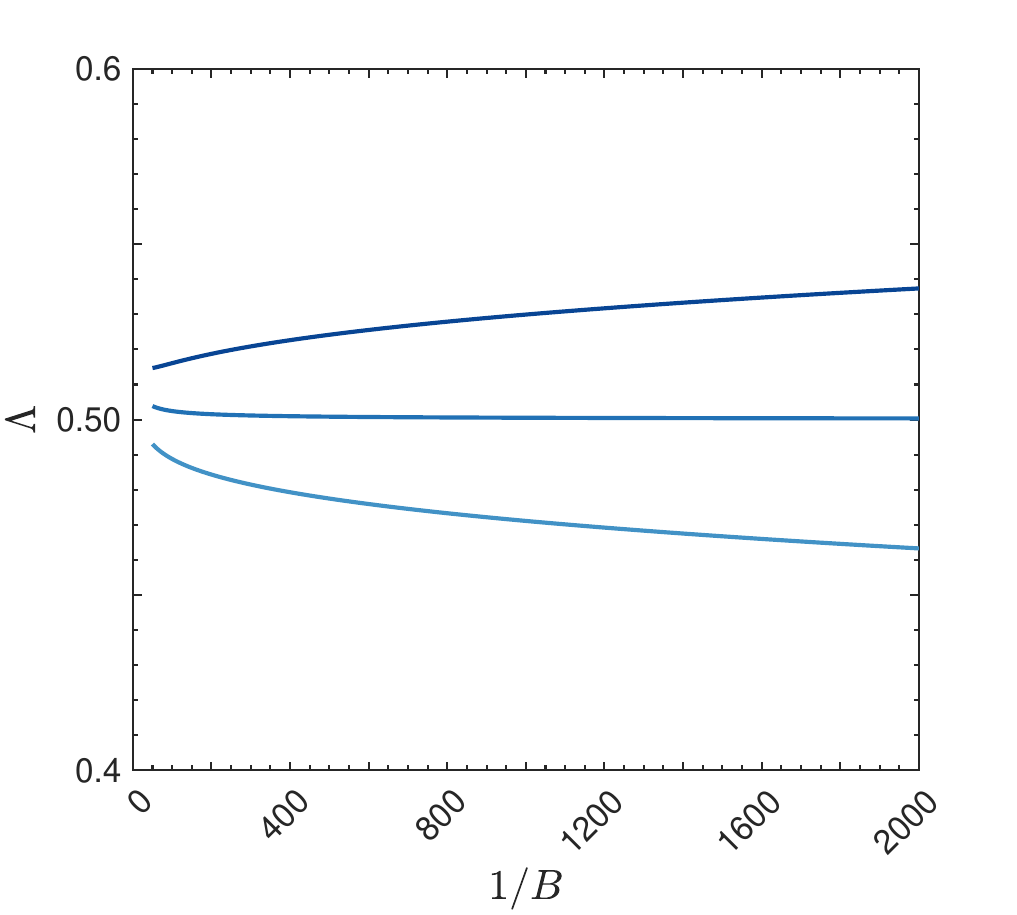} \\ 
        (c) \\ 
        \includegraphics[width=0.49\textwidth]{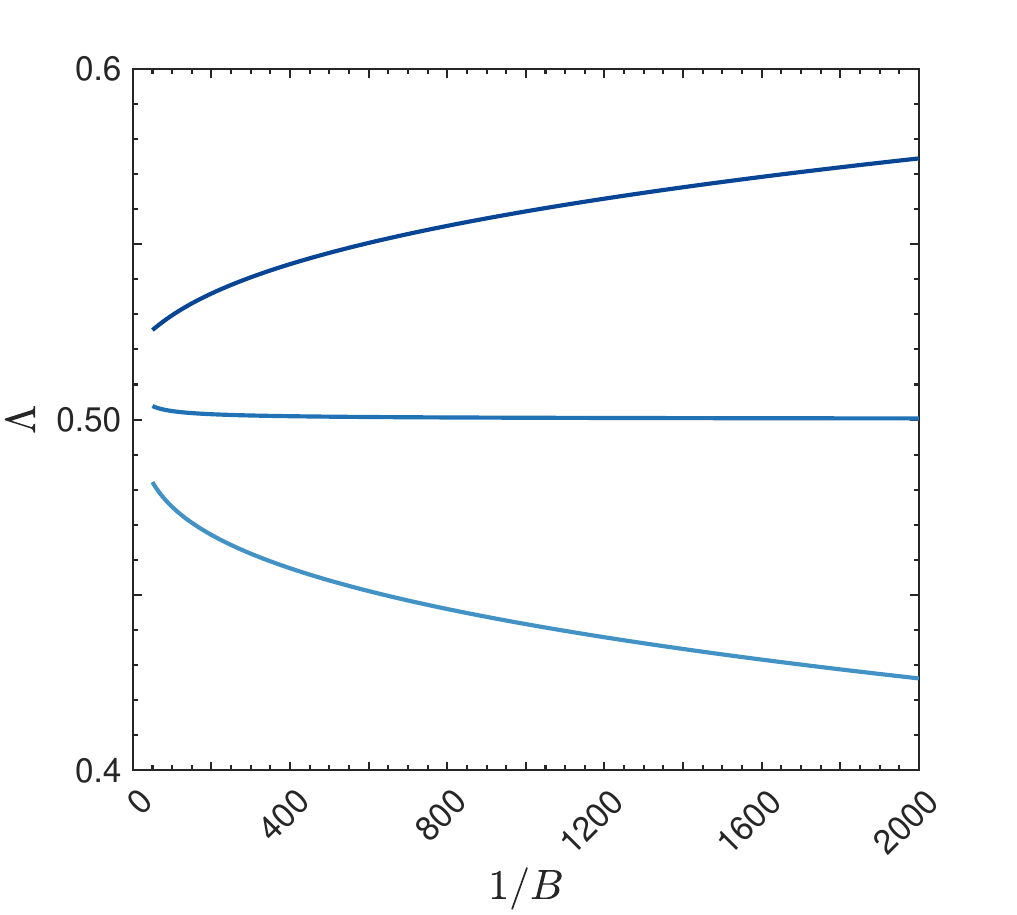}
        \caption{Dependence of $\Lambda$ on the surface tension parameter $1/B$ at different time (a) $t = 10$, (b) $t = 50$, and (c) $t = 100$. For a unique $1/B$, we obtain a unique $U$ that determines the location of the fingertip. This clearly indicates that for a given value of $1/B$, as the finger propagates along the cell, its width increases/decreases as the cell diverges/converges; whereas, the width remains almost half of the cell width for a uniform HS cell, irrespective of the tip position.}
        \label{fig:7}
\end{figure}

A qualitative comparison between the theoretical estimates of the present study and related experimental studies \cite{zhao1992perturbing, al2012control} and linear stability analysis is discussed below. Using normal mode analysis, one can show that the width of the most unstable mode satisfies \cite{al2013controlling}, 
\begin{align}
    \Lambda = \frac{2\pi}{\sqrt{3B}} \left[ 1 - \alpha \left( \frac{w}{h_0} \right)^2B \right], \label{eqn:88}
\end{align}
where 
\begin{equation}
    \label{eq:B}
    1/B = (12\mu U/\nu) [w/h_0]^2
\end{equation} 
is the surface tension parameter \cite{corvera1995anisotropic} (also called the control parameter \cite{lindner2000viscous}). Eq. \eqref{eqn:88} from the linear stability theory indicates that for the leading order term, we do not observe a significant effect of slope from the linear theory. However, a close look into the relation indicates the selected finger width $(\Lambda)$ increases/decreases for a converging/diverging Hele-Shaw cell. Our theoretical findings are in qualitative agreement with these observations (see figure \ref{fig:analytic}). In the present study, we have considered the following experimental parameters from \cite{al2012control}. The width of the Hele-Shaw cell is $W = 5.1$ cm, depth gradients $\alpha = -0.0027, 0$, $+0.0027$, and $h_0 \approx 1.4$ mm. 

Eq. \eqref{eq:B} indicates that when the geometric parameters of the Hele-Shaw cell are fixed, a unique $1/B$ corresponds to a unique $U$ that further determines the location of the fingertip ($x_0 = U t$). From figures \ref{fig:analytic} and \ref{fig:7}, we observe that for a given value of $1/B$, as the finger propagates along the cell, its width increases/decreases as the cell diverges/converges; whereas, the width remains almost half of the cell width for a uniform HS cell, irrespective of the tip position. This result for $\alpha = 0$ is consistent with the analytical data of Hong and Langer \cite{hong1986analytic} and the linear stability results of Homsy \cite{homsy1987viscous}. 
Our analytical results highlight how the small depth gradient modifies the selected finger width to explain a flatter and unstable tip for $\alpha > 0$ and a sharper and more stable tip for $\alpha < 0$ \cite{zhao1992perturbing, al2013controlling}. 
These qualitative behaviours captured through our theoretical calculation are in excellent agreement with the experimental results. 
For further quantitative comparison with the experiments, a full solution/numerical solution of the Eqs. \eqref{eqn:24}-\eqref{eqn:25} is essential, and the same will be considered in our future work. 






\appendix

\section{Mathematical re-examination of solvability theory}\label{sec:re-exam}

The Saffman-Taylor viscous fingering problem presents significant numerical challenges due to its inherently nonlinear and nonlocal nature. These complexities have led to notable discrepancies between theoretical predictions and numerical results. Even when using simplified approaches, such as the geometrical model, which neglects nonlocality, there has been persistent disagreement regarding the prefactors of the cusp function. This highlights a clear need to carefully re-examine and systematically improve the mathematical approximations made in the original proposals of solvability theory.

The purpose of this appendix is to undertake such a re-examination. We will focus on a particular part of the mathematical framework of solvability theory within the context of Hele-Shaw flow, where the complete nonlocal model must be solved. Our primary goal is to recompute the prefactor of the cusp function. In contrast to the original approach \cite{hong1986analytic}, which only included the leading term for simplicity, our calculation will incorporate all nonlinear terms at the level of the WKB approximation \eqref{eqn:wkb}. By doing so, we aim to provide a more accurate and robust framework that better aligns theoretical models with numerical observations. To include the whole series of the WKB expansion \eqref{eqn:wkb}, we propose a solution for the null eigenvector of ${\mathcal{L}}^\dagger$, of the form
\begin{align}
    \Theta_0 = e^{S_0/ \sqrt{Ca_m}}\sum_{n=0}^{\infty} g_n {Ca^{\frac{n}{2}}_m}
    \label{eqn:A1}
\end{align}
This is equivalent to \eqref{eqn:wkb}. Substituting $\Theta_0$ into \eqref{eqn:64} we obtain
\begin{gather}
    (S_0'^2 + \bar{Q})g_0 +\left[ (S_0'^2+\bar{Q})g_1 +(2S_0'g_0'+S_0''g_0) \right] Ca_m^{1/2} \nonumber\\ ~~~~~~~~~~~~~~~~~~~~~~~~~~~~~~~~~+\sum_{m=0}^{\infty}\left[ g_m'' + (S_0'^2+\bar{Q})g_{m+2} + (2S_0'g_{m+1}'+S_0''g_{m+1}) \right]Ca_m^{\frac{m}{2}+1}=0
    \label{A2}
\end{gather}
Equating terms with the same power in $Ca_m$, we have
\begin{enumerate}[label=(\alph*)]
    \item $(S_0'^2 + \bar{Q})g_0 =0$
    \label{item:a}
    \item $(S_0'^2+\bar{Q})g_1 +(2S_0'g_0'+S_0''g_0) =0$
    \label{item:b}
    \item $ g_m'' + (S_0'^2+\bar{Q})g_{m+2} + (2S_0'g_{m+1}'+S_0''g_{m+1})=0$
    \label{item:c}
\end{enumerate}
which is a set of equations for $S_0, g_0, g_1,~.~.~.~,g_n$. From \ref{item:a}, we have $S_0'^2 + \bar{Q}=0$, so
\begin{align}
    S_0 = i \int_{0}^{\eta} \bar{Q}^{1/2} d\eta
    \label{eqn:A3}
\end{align}
and
\begin{enumerate}[label=(\alph*), start=2]
    \item $g_0' +\frac{S_0''}{2S_0'}g_0 =0$
    \label{item:b'}
    \item $g_{m+1}'+\frac{S_0''}{2S_0'}g_{m+1}= -\frac{1}{2S_0'}g_m''$.
    \label{item:c'}
\end{enumerate}
Since $S_0$ has a point of stationary phase at $\bar{\eta}=i$, we calculate our equations in the immediate neighbourhood of $\bar{\eta}$. Let $\eta=i+\omega$, then
\begin{gather*}
    \bar{Q} \simeq b_1 \omega^{3/2}, ~\text{where}~ b_1 \equiv \frac{2^{5/2}i^{3/2}}{(1+\alpha \frac{x_0}{h_0})^3}\left( \frac{\Lambda^2}{1-2\Lambda} \right)^2 ,\\
    \frac{S_0''}{2S_0'} = \frac{1}{4}\frac{\bar{Q}'}{\bar{Q}} \simeq \frac{3}{8} \frac{1}{\omega}
\end{gather*}
and from $S_0'^2 + \bar{Q}=0$ will give $2S_0' = 2ib_1^{1/2}\omega^{3/4}$.

The equation around $\bar{\eta}$ become
\begin{enumerate}[label=(\alph*), start=2]
    \item $g_0' +\frac{3}{8} \frac{1}{\omega}g_0 =0$
    \label{item:b''}
    \item $g_{m+1}'+\frac{3}{8} \frac{1}{\omega}g_{m+1}= -\frac{1}{2ib_1^{1/2}\omega^{3/4}}g_m''$.
    \label{item:c''}
\end{enumerate}
From \ref{item:b''} we obtain 
\begin{align}
    g_0 = a_0 \omega^{-3/8},
    \label{eqn:A4}
\end{align}
and \ref{item:c''} gives a recursive relation for the rest of the $g_n$'s. It is actually a first order equation of the form $g' + Fg = G$ whose solution is $g = g_h \int \frac{G}{g_h} d\omega$ where $g_h = e^{-\int F d\omega}$. The solution for $g_{m+1}$ is then given in terms of $g_m$ by 
\begin{align}
    g_{m+1} = \omega^{-A}B \int \omega^{A-\frac{3}{4}}g_m'' d\omega
    \label{eqn:A5}
\end{align}
with $A= 3/8$ and $B=-1/2ib_1^{1/2}$. It is clear that if $g_m$ has the form $g_m= a_m \omega^{-A_m}, g_{m+1}$ will have the same form. And since $g_0$ has the form $g_0 = a_0 \omega^{-A_0}$, with $A_0 = 3/8$. All we need to do is to find a recursive relation among the coefficients. When the form $g_m = a_m \omega^{-A_m}$ is put into \eqref{eqn:A5} we find
\begin{align}
    g_{m+1} = \frac{BA_m(A_m+1)a_m}{A-A_m-\frac{7}{4}}\omega^{-(A_m+\frac{7}{4})} = a_{m+1}\omega^{-A_{m+1}}.
    \label{eqn:A6}
\end{align}
Comparison of the exponents of $\omega$ determines $A_m$
\begin{align}
    A_{m} = \frac{7}{4}m + A_0
    \label{eqn:A7}
\end{align}
and comparison of the multiplicative coefficients determines $a_m$
\begin{align}
    a_{m} = \left(-\frac{7}{4} \right)^m B^m \frac{\Gamma(m+\frac{3}{14})\Gamma(m+\frac{11}{14})}{\Gamma(\frac{3}{14})\Gamma(\frac{11}{14})\Gamma(m+1)}a_0 .
    \label{eqn:A8}
\end{align}
For consistency, we need to expand $e^{\frac{S_0}{\sqrt{Ca_m}}}$ and $\bar{R}$ (expression of $\bar{R}$ in Eqn. \eqref{eqn:55}) around the point of stationary phase $\bar\eta$. The subsequent analysis depends on the choice of $\Lambda > 1/2$ and $\Lambda < 1/2$, which are discussed in the following subsections. 


\subsection{$\Lambda > 1/2$}

Using equations \eqref{eqn:84} and \eqref{eqn:A3}, $S_0$ can be approximated as
\begin{align}
    S_0(\eta) \simeq E(\Lambda,\alpha) + \frac{2\pi i}{(1+\alpha \frac{x_0}{h_0})^{3/2}} \frac{(2\Lambda -1)^{3/4}}{(1-\Lambda)} + \frac{8}{7} \frac{2^{1/4}i^{7/4}}{(1+\alpha \frac{x_0}{h_0})^{3/2}}\frac{\Lambda^2}{1-2\Lambda}\omega^{7/4}.
    \label{eqn:A9}
\end{align}
where $\omega = \eta -i$ and $E(\Lambda, \alpha)$ is given by \eqref{eqn:81}. We can write
\begin{align*}
    e^{\frac{S_0}{\sqrt{Ca_m}}} = e^{\frac{E(\Lambda,\alpha)}{\sqrt{Ca_m}}}e^{\frac{a\omega^{7/4}}{\sqrt{Ca_m}}} \cos \left[ \frac{2\pi}{(1 + \alpha \frac{x_0}{h_0})^{3/2}} \frac{(2\Lambda -1)^{3/4}}{(1-\Lambda)\sqrt{Ca_m}} \right] ,\\ a \equiv  \frac{8}{7}\frac{2^{1/4}i^{7/4}}{(1+\alpha \frac{x_0}{h_0})^{3/2}}\frac{\Lambda^2}{1-2\Lambda} .
\end{align*}
With the above results \eqref{eqn:A1} becomes
\begin{align}
    \Theta_0 =  e^{\frac{E(\Lambda,\alpha)}{\sqrt{Ca_m}}} \cos \left[ \frac{2\pi}{(1 + \alpha \frac{x_0}{h_0})^{3/2}} \frac{(2\Lambda -1)^{3/4}}{(1-\Lambda)\sqrt{Ca_m}} \right] \sum_{m=0}^{\infty} a_m Ca_m^{\frac{m}{2}} \omega^{-\frac{7}{4}m-\frac{3}{8}} e^{\frac{a\omega^{7/4}}{\sqrt{Ca_m}}}
    \label{eqn:A10}
\end{align}
and $\bar{R}$, given by equation \eqref{eqn:55}, becomes
\begin{align*}
    \bar{R} = b_2 \omega^{-9/4}, ~~\text{where}\\
    b_2 \equiv \frac{3i^{-5/4}}{2^{9/4}}\frac{(1-2\Lambda)^{1/2}}{(1-\Lambda)}.
\end{align*}
We can now write the cusp function as 
\begin{align}
    \mathcal{C} = \int_{-\infty}^{\infty} \Im \Theta_0 \bar{R}(\omega) d\omega = \frac{1}{i} b_2 e^{\frac{E(\Lambda,\alpha)}{\sqrt{Ca_m}}} \cos \left[ \frac{2\pi}{(1 + \alpha \frac{x_0}{h_0})^{3/2}} \frac{(2\Lambda -1)^{3/4}}{(1-\Lambda)\sqrt{Ca_m}} \right] \sum_{m=0}^{\infty} a_m Ca_m ^{\frac{m}{2}}I_m , 
    \label{eqn:A11}
\end{align}
where $I_m$ is defined as 
\begin{align*}
    I_m \equiv \int_{-\infty}^{\infty} \omega^{-\frac{7}{4}(m+\frac{3}{2})}e^{\frac{a\omega^{7/4}}{\sqrt{Ca_m}}} ,
\end{align*}
which can be written in terms of a gamma function as 
\begin{align*}
    I_m = \frac{8\pi i}{7} \left( \frac{a}{\sqrt{Ca_m}} \right)^{m+\frac{13}{14}} \frac{1}{\Gamma(m+\frac{27}{14})}.
\end{align*}
The cusp function can then be written as
\begin{align}
\label{eqn:cusp}
    \mathcal{C} = \left[ \frac{8\pi}{7}a_0 b_2 a^{\frac{13}{14}} \frac{1}{Ca_m^{13/28}} \right] e^{\frac{E(\Lambda,\alpha)}{\sqrt{Ca_m}}} \cos \left[ \frac{2\pi}{(1 + \alpha \frac{x_0}{h_0})^{3/2}} \frac{(2\Lambda -1)^{3/4}}{(1-\Lambda)\sqrt{Ca_m}} \right] \Delta ,
\end{align}
where
\begin{align}
    \Delta = \frac{1}{a_0} \sum_{m=0}^{\infty} a^m a_m \frac{1}{\Gamma(m+\frac{27}{14})} = \frac{1}{\Gamma(\frac{3}{14})\Gamma(\frac{11}{14})} \sum_{m=0}^{\infty} \left[\frac{1}{2}\right]^m \frac{\Gamma(m+\frac{3}{14})\Gamma(m+\frac{11}{14})}{\Gamma(m+1)\Gamma(m+\frac{27}{14})}.
    \label{eqn:A12}
\end{align}
Equation \eqref{eqn:A8} has been used to write the second equality. There is an overall constant which is irrelevant since what is uniquely determined is the normalized cusp function. To compare with the original calculation, we set $a_0$ to agree with the result of Hong and Langer \cite{hong1986analytic} when only the term $m=0$ in \eqref{eqn:A12} is considered. We have
\begin{align*}
    \mathcal{C} = \frac{N}{(1+\alpha\frac{x_0}{h_0})^{9/14}} \frac{(1-2\Lambda)^{1/14}\Lambda^{6/7}}{(1-\Lambda)}\frac{1}{Ca_m^{13/28}}e^{\frac{E(\Lambda,\alpha)}{\sqrt{Ca_m}}} \cos \left[ \frac{2\pi}{(1 + \alpha \frac{x_0}{h_0})^{3/2}} \frac{(2\Lambda -1)^{3/4}}{(1-\Lambda)\sqrt{Ca_m}} \right] \Delta \Gamma\left(\frac{27}{14}\right),
\end{align*}
where
\begin{align*}
    N = 3\pi \frac{2^{22/7}}{7^{27/14}\Gamma(\frac{27}{14})} = 2.008 .
\end{align*}
Hence, the quantity $\Delta \Gamma\left(\frac{27}{14}\right)$ provides a multiplicative factor that indicates how distinct this result is from the result where only two terms in Eqn. \eqref{eqn:wkb} are kept. So, when setting $m = 0, \Delta \Gamma\left(\frac{27}{14}\right)=1$ and the result of Eqn. \eqref{eqn:87} is recovered. 


\subsection{$\Lambda < 1/2$}

Similar to the case of $\Lambda > 1/2$, using equations \eqref{eqn:76} and \eqref{eqn:A3}, $S_0$ can be approximated as
\begin{align}
    S_0 \simeq E(\Lambda, \alpha) + \frac{8}{7}\frac{2^{1/4}i^{7/4}}{(1+\alpha \frac{x_0}{h_0})^{3/2}}\frac{\Lambda^2}{1-2\Lambda}\omega^{7/4}. 
    \label{eqn:A14}
\end{align}
Note that the expression for $S_0$ in \eqref{eqn:A9} corresponding to $\Lambda > 1/2$ has an additional term 
\[ \frac{2\pi i}{(1+\alpha \frac{x_0}{h_0})^{3/2}} \frac{(2\Lambda -1)^{3/4}}{(1-\Lambda)} \]
that subsequently appears in the cusp function (see equation \eqref{eqn:cusp}) as the argument of a cosine function. In the case of $\Lambda < 1/2$, the corresponding cosine term in the cusp function becomes $1$. Therefore, the cusp function reads 
\begin{align}
\label{eq:cusp2}
    \mathcal{C} = \frac{N}{(1+\alpha\frac{x_0}{h_0})^{9/14}} \frac{(1-2\Lambda)^{1/14}\Lambda^{6/7}}{(1-\Lambda)}\frac{1}{Ca_m^{13/28}}e^{\frac{E(\Lambda,\alpha)}{\sqrt{Ca_m}}} \Delta \Gamma\left(\frac{27}{14}\right),
\end{align}
where
\begin{align*}
    N = 3\pi \frac{2^{22/7}}{7^{27/14}\Gamma(\frac{27}{14})} = 2.008 .
\end{align*}
Similar to the case of $\Lambda > 1/2$, the quantity $\Delta \Gamma\left(\frac{27}{14}\right)$ gives a multiplicative factor, which tells us how much this result differs as compared to the result with only two terms in Eqn. \eqref{eqn:wkb}. Therefore, setting $m = 0$, one obtains $\Delta \Gamma\left(\frac{27}{14}\right) = 1$ and Eqn. \eqref{eqn:78} is recovered. 

Therefore, for both $\Lambda > 1/2$ and $\Lambda < 1/2$, we note, $N' = M\Delta \Gamma\left(\frac{27}{14}\right)$. The factor $\Delta$ can be computed numerically to obtain $\Delta \Gamma\left(\frac{27}{14}\right) = 1.054$
\cite{corvera1995anisotropic}. Thus, including all terms in the WKB expansion leads to approximately $5\%$ change in the pre-factor.

\section*{Acknowledgement} D.G. acknowledges the financial support from IPDF, IIT Guwahati, through project grant no. MATHSPNIITG01371xxDG001. S.P. acknowledges financial support through the MATRICS Grant (MTR/2022/000493) from the Science and Engineering Research Board, Department of Science and Technology, Government of India.

\section*{Conflict of interest statement} All the authors certify that they do not have any conflict of interest.


\bibliographystyle{unsrt} 
\bibliography{sample.bib}

\end{document}